\documentclass[10pt, conference, letterpaper]{IEEEtran}
\IEEEoverridecommandlockouts
\usepackage{amsmath,amssymb,amsfonts}
\usepackage{algorithmic}
\usepackage{graphicx}
\usepackage{subfig}
\usepackage{textcomp}
\usepackage{xcolor}
\usepackage{multirow}
\usepackage{booktabs}
\usepackage[ruled,vlined,linesnumbered]{algorithm2e}
\usepackage[numbers,sort&compress]{natbib}
\usepackage{hyperref}
\usepackage{bm}
\usepackage{balance}
\usepackage{amsthm}
\newcommand{\paratitle}[1]{\vspace{1.1ex}\noindent\textbf{#1}}
\newtheorem{assumption}{Assumption}
\newtheorem{theorem}{Theorem}

\newtheorem{lemma}{Lemma}

\newcommand{\noaistats}[1]{}
\newcommand{\SUB}[1]{\ENSURE \hspace{-0.3in} \textbf{#1}}

\def\BibTeX{{\rm B\kern-.05em{\sc i\kern-.025em b}\kern-.08em
    T\kern-.1667em\lower.7ex\hbox{E}\kern-.125emX}}

\SetKwInput{KwInput}{Input}                
\SetKwInput{KwOutput}{Output}              
\SetKwComment{Comment}{$\triangleright$\ }{}

\begin{document}
\title{Learnable Sparse Customization in Heterogeneous Edge Computing}

\author{\IEEEauthorblockN{Jingjing Xue$^{1, 2}$, Sheng Sun$^{1}$, Min Liu$^{1,2,5,*}$\thanks{* Min Liu is the Corresponding author.}, Yuwei Wang$^{1}$, Zhuotao Liu$^{3,5}$, Jingyuan Wang$^{4}$}
\IEEEauthorblockA{$^{1}$Institute of Computing Technology, Chinese Academy of Sciences, Beijing, China \\
$^{2}$University of Chinese Academy of Sciences, Beijing, China \\
$^{3}$Institute for Network Sciences and Cyberspace, Tsinghua University, Beijing, China \\
$^{4}$School of Computer Science and Engineering, Beihang University, Beijing, China \\
$^{5}$Zhongguancun Laboratory, Beijing, China\\
\IEEEauthorblockA{Email: \{xuejingjing20g, sunsheng, liumin, ywwang\}@ict.ac.cn, zhuotaoliu@tsinghua.edu.cn, jywang@buaa.edu.cn}
}}

\maketitle
\thispagestyle{plain}

\begin{abstract}
To effectively manage and utilize massive distributed data at the network edge, Federated Learning (FL) has emerged as a promising edge computing paradigm across data silos. However, FL still faces two challenges: system heterogeneity (\emph{i.e.}, the diversity of hardware resources across edge devices) and statistical heterogeneity (\emph{i.e.}, non-IID data). Although sparsification can extract diverse submodels for diverse clients, most sparse FL works either simply assign submodels with artificially-given rigid rules or prune partial parameters using heuristic strategies, resulting in inflexible sparsification and poor performance. In this work, we propose Learnable Personalized Sparsification for heterogeneous Federated learning (FedLPS), which achieves the learnable customization of heterogeneous sparse models with importance-associated patterns and adaptive ratios to simultaneously tackle system and statistical heterogeneity. Specifically, FedLPS learns the importance of model units on local data representation and further derives an importance-based sparse pattern with minimal heuristics to accurately extract personalized data features in non-IID settings. Furthermore, Prompt Upper Confidence Bound Variance (P-UCBV) is designed to adaptively determine sparse ratios by learning the superimposed effect of diverse device capabilities and non-IID data, aiming at resource self-adaptation with promising accuracy. Extensive experiments show that FedLPS outperforms status quo approaches in accuracy and training costs, which improves accuracy by 1.28\%-59.34\% while reducing running time by more than 68.80\%.
\end{abstract}


\section{Introduction}

The proliferation of mobile and IoT devices drives the significant growth of data generated at the network edge \cite{ijcai2024p245, 9385890}, which, coupled with the demand for  real-time data processing and privacy protection, has fueled the rise of edge data management and utilization \cite{Frontier_EDM_2018, EDM_2020, data_mana_2023}. Traditional centralized methods require edge devices to upload such a huge amount of data for centralized processing, which tends to exhaust network capacity and bring unacceptable transferring latency \cite{edge_intelligence_2021, FedMigr_2022}. Besides, raw data uploading takes the risk of user privacy leakage and unauthorized access \cite{FedADMM_2022, FedCross_2024}. Faced with these limitations, the landscape of data management has significantly changed, giving rise to embracing Mobile Edge Computing (MEC) in distributed data management across edge users \cite{EC_DM_2022}. MEC pushes data storage and model computing to network edges at the source instead of raw data transmission, enabling real-time processing and safeguarding privacy. As a distributed edge computing paradigm, FL \cite{FL_MEC_2020, Fed_Survey_2023, FL_Survey_2024} allows devices to locally process data and jointly compute a shared model without raw data sharing, which has become a promising solution for privacy-preserving edge data management and utilization \cite{FL_DM_2023, FedMix_2024}. In FL, edge devices (\emph{i.e.}, clients) locally update models on their own data and periodically upload local updates to the server for aggregating into a global model.

Despite the benefits of data localization and privacy protection, FL still faces two key challenges:
(1) \textbf{System heterogeneity} highlights that different edge devices possess diverse resource configurations, which limits clients to train models matching their capabilities, resulting in severe performance gaps \cite{HeteroFL_Survey_2024, FedLMT_2024}. (2) \textbf{Statistical heterogeneity} focuses on non-IID data among clients \cite{non-IID_2024, FedAPEN_2023}, which leads to apparent inconsistencies between local updates \cite{inconsistency_update}, further degrading model generalization \cite{FedRep_2021, FedCP_2023}. The simultaneous existence of such dual heterogeneity brings additive bottlenecks in processing efficiency and performance, hindering the deployment of FL in practical Edge Data Management (EDM) and MEC scenarios.

Prior works prioritize clients with powerful capabilities \cite{Oort_2021} and tolerate stale updates \cite{REFL_2023}, which introduces training bias due to poor fairness and sacrifices accuracy considering global model drift. While some studies \cite{FedProx_2020, FedDyn_2021} limit the inconsistencies between local updates and global models to ensure convergence on non-IID data, which still suffer from heavy burden. The primary drawback of both directions is the identical model architecture across all clients, which incurs stragglers slowing down the FL process and brings inference accuracy gaps. Therefore, different local models are advocated to match heterogeneous settings of edge devices \cite{FedRolex_2022}. 

Sparsification is a promising solution for heterogeneous model extraction, which can prune different parameters from the global model to build diverse submodels. A sparse model involves two determining factors: (1) \textbf{Sparse ratio} indicates how many parameters are retained after sparsification. (2) \textbf{Sparse pattern} points out which parameters are removed, which determines the submodel structure. It has been verified that changes in both sparse ratio and pattern yield notable differences in resource costs and model accuracy \cite{FedMP_2022, FedBIAD_2023, FlexFL_2024}.

\begin{figure}
	\centering
	\includegraphics[width=0.92\columnwidth]{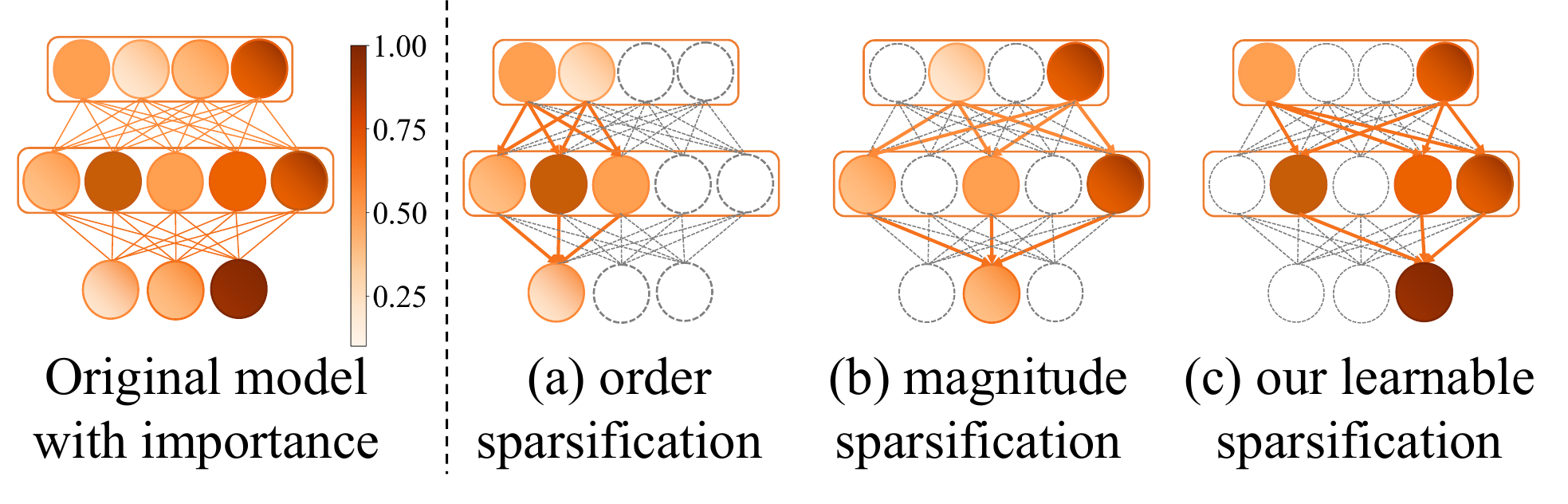}
	\caption{Different pattern strategies. The padding represents importance scores and no padding indicates the unit is sparsified.}
	\label{diff_sparsification}
	\vskip -0.15in
\end{figure}

Several pioneering studies \cite{HeteroFL_2021,FlexFL_2024} introduce sparsification into FL to assign distinct models for clients, where rigid rules of sparse ratio are artificially set based on device capabilities. However, sparse ratios are not only restricted by diverse capabilities, but also closely correlated with model accuracy \cite{FedMP_2022, FedDrop_2022}. These rigid rules cannot flexibly model the interaction between sparse ratio, capability, and accuracy, failing to trade off resource self-adaption and accuracy guarantee over non-IID data. 
To fit non-IID data, recent works \cite{Hermes_2021, FedSpa_2022, FedP3_2024} explore the personalization of sparse models. They heuristically specify sparse patterns including random \cite{FD_2018}, ordered \cite{Fjord_2021}, and magnitude-based \cite{FedSpa_2022, CS_2023} sparsify, ignoring the precise measure of model unit importance on local data. As shown in Figure \ref{diff_sparsification}, we use the depth of padding to indicate model unit importance, and these heuristic strategies prune some significant units. The resulting sparse model loses the representation abilities of the corresponding important parameters and cannot accurately extract local data features, resulting in performance degradation. Overall, artificially-given rigid ratio rules and heuristic pattern strategies cannot effectively accommodate complicated dual heterogeneity in real EDM scenarios.

In light of the above observation, we propose Learnable Personalized Sparsification for heterogeneous Federated learning (FedLPS), which enables learnable customization of sparse patterns and sparse ratios to tailor a capability-affordable and accuracy-remarkable submodel for each client. Specifically, FedLPS links the model unit importance and local data in local loss. With the importance-associated loss, the client accurately learns importance indicators based on local data by the back propagation to derive importance-based sparse patterns with minimal heuristics, achieving personalized sparse training. Furthermore, we design P-UCBV with accuracy-dominated arm elimination to learn the correlation among sparse ratio, capability, and accuracy. Based on the superposition effect of diverse capabilities and non-IID data, we adaptively determine sparse ratios for clients to flexibly accommodate all possible cases. In this way, FedLPS allows each client to customize a sparse model with the learnable pattern and adaptive ratio to accurately process personalized data and match diverse capabilities, boosting model performance and computing efficiency. 

Our main contributions can be summarized as follows:
\begin{itemize}
\item  FedLPS takes the first step to customize a data-driven and resource-adapted sparsification in learnable ways for each client so as to accelerate the training process while enhancing inference accuracy in complicated non-IID and system-heterogeneous MEC scenarios. 
\item For precise and efficient edge data processing, unit-wise importance indicators are optimized on local data to facilitate learnable pattern personalization, and P-UCBV learns additive feedback of time costs and accuracy under diverse resource and non-IID data to adaptively determine computation-efficient and accuracy-guaranteed ratios.
\item Theoretically, we provide the upper bound of the gap between global and local parameters in heterogeneous settings and prove the convergence of FedLPS under SGD optimization with the constraints of learning rates.
\item We conduct experiments on 5 classic datasets with various models. The results show that FedLPS provides 1.28\%-59.34\% accuracy gains while achieving up to 60\% reduction in computation costs and more than 68.80\% time acceleration compared to baselines.

\end{itemize}

\section{Preliminaries and Problem Formulation}

Suppose $K$ edge devices with data $\{D_1, \ldots, D_K\}$ participate in edge computing. These distributed data are essentially non-IID and devices generally hold heterogeneous resource configurations in real MEC scenarios. Let $z_k$ represent the computation capability of client $k$.

\paratitle{Personalization Target.} To fit non-IID data, the personalization paradigm learn an individual model for each client. Let $\omega_k$ denote local model parameters of client $k$. The local mask $m_k \in \{0,1\}^{|\omega_k|}$ induces sparsification by $\omega_k \odot m_k$, where $\odot$ is a Hadamard product. Our target is to optimize local sparse models to minimize the average loss of clients, expressed as
\begin{equation}
	\label{PFL_paradigm}
	\begin{split}
		& \min_{\omega_1, \ldots, \omega_K \in \Omega}{\frac{1}{K}\sum_{k=1}^{K} F_k(\omega_k \odot m^*_k )}, \\
		\text{s.t.} \quad  & F_k(\omega_k \odot m^*_k) =  \mathbb{E} [ \mathcal{L}_k(\omega_k \odot m^*_k \, | \, \omega; D_k)]
	\end{split}
\end{equation}
Here, $\Omega$ is the feasible parameter space of local models, and $\mathcal{L}_k(\cdot\, | \omega; D_k)$ denotes the regularization loss of client $k$ over $D_k$ under global parameters $\omega$. For local mask $m_k$, if an element of $m_k$ is zero, the corresponding parameters are zeroed out, otherwise remain active. In this setting, the nonzero parameters of $\omega_k \odot m_k$ can characterize the submodel. The nonzero parameters are trained and uploaded to the server, meaning that the local running burden and uplink communication volume can be reduced compared to the original dense model.
The optimal mask of client $k$ is denoted as $m^*_k$, which is hard to learn in a non-heuristic way.

\begin{figure*}
	\centering
	\includegraphics[width=1.8\columnwidth]{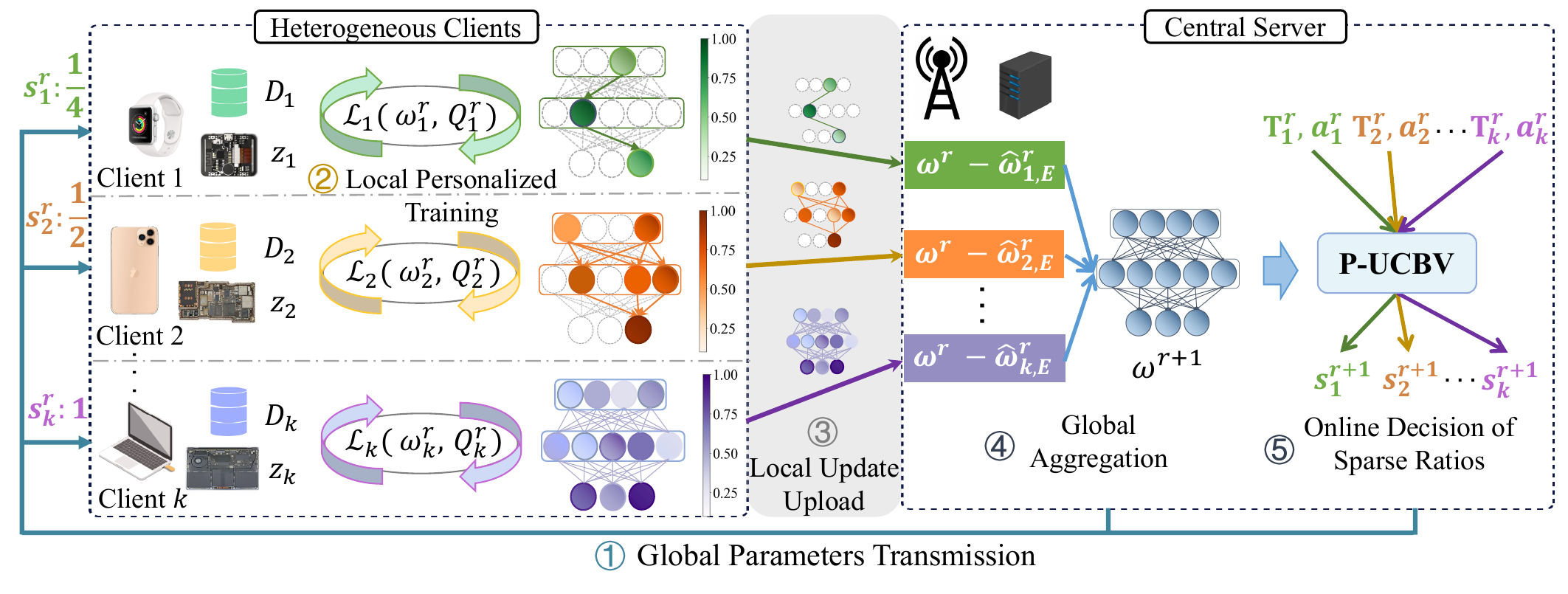}
	\caption{The overview diagram of our FedLPS framework, where the numbers \normalsize{\textcircled{\scriptsize{1}}}\normalsize \--{}\normalsize{\textcircled{\scriptsize{5}}}\normalsize\, represent the training procedures.}
	\label{overview}
	\vskip -0.15in
\end{figure*}

\paratitle{Sparse Ratio and Pattern Definition.} Let $N_k$ denote the number of nonzero elements in $m_k$ such that $N_k=\|m_k\|_0$. We define the sparse ratio $s_k$ of client $k$ as $s_k = N_k/|\omega_k|$, where a lower ratio implies a higher sparse degree. Given $s_k$, the local mask $m_k$ is determined by the sparse pattern $P_k$ that indicates the positions of the retained model units. In view of general training speed capabilities, structured sparsification is our primary focus, which considers the structurally indivisible element (\emph{e.g.}, neuron and convolution channel) in DNN as the sparse granularity. We define a network topology element at the sparse granularity level as a model unit.
With unit-wise sparse pattern $P_k$, the local mask is derived by
\begin{equation}
\label{mask_derive}
m_k = \mathcal{M}(P_k\; |\; \omega_k, s_k),
\end{equation}
where $\mathcal{M}(P|\omega, s)$ denotes the construction of a binary mask with the same dimension as $\omega$ and sparse ratio $s$ by setting 1 in the positions where $P$ is true and 0 elsewhere.  Several studies \cite{AFD_2020, FedMP_2022, FedBIAD_2023} demonstrate that sparse patterns and ratios have significant impacts on running costs and model performance. In view of non-IID data and diverse resource configurations across edge devices, it is crucial to learn personalized sparse patterns and adaptive sparse ratios, aiming at local resource self-adaption and accuracy improvement in edge computing.

\section{The FedLPS Framework}

\subsection{Overview}

To deal with complicated issues introduced by non-IID data and system heterogeneity in edge data management and computing scenarios, we propose a learnable sparse customization framework, FedLPS, as illustrated in Figure \ref{overview}. The core components of FedLPS include (1) \emph{customizing a learnable sparse pattern} by optimizing the importance scores of model units during the back-propagation of local training and (2) \emph{determining an adaptive sparse ratio} through learning the additive feedback of client-specific resource constraints and accuracy changes over local data. Specifically, the procedures of FedLPS in a communication round involve:
\begin{itemize}
\item \textbf{Local personalized training with learnable sparse patterns:} Each client maintains a unit-wise importance indicator to derive a personalized sparse pattern. The client initializes the local model and sparsifies it with the sparse ratio and pattern. With designed importance-associated loss, local model parameters and importance-based sparse patterns are updated over local data via back propagation. After local training, the client uploads sparse update, local cost, and accuracy to the server.
\item \textbf{Global aggregation and sparse ratio decision:} The server aggregates local updates and determines sparse ratios. The online decision of sparse ratios is viewed as a MAB problem. Based on local cost and accuracy statistics, P-UCBV is developed to adaptively select an appropriate sparse ratio for the client.
\end{itemize}
The above steps are cyclically iterated until round $r \geq R$. The whole process of FedLPS is summarized in Algorithm \ref{FedLPS}. 

\subsection{Learnable Sparse Training to Personalization}

\paratitle{Importance Indicator.} 
In non-IID settings, the same model unit exhibits distinct importance on different clients. We introduce an importance indicator for each client to individually measure the significance of each model unit over local data. We define $Q_{k, l}^r$ as the local importance indicator in the $l$-th iteration of round $r$ on client $k$, which can be expanded to
\begin{equation}
	Q_{k, l}^r = [q_{k,l}^{r,1}, \ldots, q_{k,l}^{r,J}]^\top \in \mathbb{R}^J.
\end{equation}
Here,  $J$ denotes the number of sparsifiable units in the local DNN. For instance, in a Fully-Connected Neural Network (FCNN), the sparsifiable units are neurons and $J$ is the number of neurons. Each element $q_{k,l}^{r,j}, \, j\in\{1,\ldots, J\}$ denotes the importance score of the $j$-th sparsifiable unit in the local model. A higher score indicates that the corresponding unit exhibits more significant representation ability for local data. The units with greater significance should be held in sparse models to effectively represent local data and improve performance.

\begin{algorithm}[htbp]
	\caption{FedLPS}
	\small
	\label{FedLPS}
	\textbf{Initialize}: client selection fraction $\epsilon$; the number of local iterations $E$; initial global parameters $\omega^0$; initial importance indicators $\{Q^s_{1}, \dots, Q^s_{K}\}$; initial sparse ratios $\{s_1^0, \ldots, s_K^0\}$;
	\begin{algorithmic}[1] 
		\SUB{ServerAction:}
		\STATE Set the number of selected clients in a round $C \leftarrow \max(\lfloor \epsilon \cdot K \rfloor ,\, 1)$
		\FOR{each round $r \in \{0,\ldots,R-1\}$}
		\STATE $\mathcal{C}_r \leftarrow$ random set of $C$ clients
		\STATE Send global parameters $\omega^{r}$ and the sparse ratio $s_k^r$ \\ to each selected client $k\in\mathcal{C}_r$
		\FOR{each client $k \in \mathcal{C}_r$ \textbf{in parallel}}
		\STATE $\hat{\omega}_{k,E}^r,\; T_k^r,\; a_k^r  \leftarrow$ \textbf{ClientUpdate}($\omega^{r}$, $s_k^r$) \\ \emph{// Local personalized training}
		\ENDFOR
		\STATE Global aggregation via Equation (\ref{aggreagtion})
		\FOR{$k\in\{1,2, \ldots, K\}$}
		\IF{$k\in\mathcal{C}_r$} 
		\STATE $s_k^{r+1} \leftarrow$ \textbf{P-UCBV}($T_k^r,\; a_k^r$)  \emph{// Online decision \\ of sparse ratios by Algorithm \ref{EUCBV}}
		\ELSE 
		\STATE $s_k^{r+1} \leftarrow s_k^r$
		\ENDIF
		\ENDFOR 
		\ENDFOR
		\SUB{}
		\SUB{ClientUpdate ($\omega^{r}$, $s_k^r$):}    \qquad\qquad \emph{// Done by client $k$}
		\STATE  Initialize $\omega^r_{k,0} \leftarrow \omega^{r}$ and 
		$Q^r_{k,0}\leftarrow Q^s_k$ 
		\FOR{each iteration $l=\{0,1,\ldots,E-1\}$}
		\STATE Sample a batch of training data $d_l^k$
		\STATE $\omega_{k,l+1}^r \leftarrow \omega_{k,l}^r  \; - \;  \textbf{SGD}\Big(\mathcal{L}_k(\omega_{k,l}^r,\; Q_{k,l}^r \; | \; \omega^r, d^k_l),$ \\ $\omega_{k,l}^r \odot \mathcal{M}\big(\gamma(Q_{k, l}^r - \tau_{k, l}^r I) | \;\omega_{k,l}^r, s_k^r \big), \eta_r \Big)$
		\STATE $Q_{k,l+1}^r = Q_{k,l}^r  -  \textbf{SGD}\Big(\mathcal{L}_k(\omega_{k,l}^r,\; Q_{k, l}^r \; | \; \omega^r, d^k_l), \; Q_{k,l}^r, \eta_r\Big)$
		\ENDFOR
		\STATE Record $Q_k^s \leftarrow Q_{k,E}^r$
		\STATE Obtain local personalized model with sparse \\ parameters $\omega_{k,E}^r \odot \mathcal{M}\big(\gamma(Q_{k, E}^r - \tau_{k, E}^r I) | \;\omega_{k,E}^r, s_k^r \big)$
		\STATE $\hat{\omega}_{k,E}^r \leftarrow (\omega^{r} - \omega_{k,E}^r) \odot \mathcal{M}\big(\gamma(Q_{k, E}^r - \tau_{k, E}^r I) | \;\omega_{k,E}^r, s_k^r \big)$
		\STATE Count local cost $T^r_k$ and average training accuracy $a^r_k$
		\STATE Return $\hat{\omega}_{k,E}^r,\; T^r_k, \; a^r_k$ to the server
	\end{algorithmic}
\end{algorithm}

\paratitle{Importance-Derived Sparse Pattern.} We can derive a binary sparse pattern $P_{k,l}^r=[\beta_{k,l}^{r,1}, \ldots, \beta_{k,l}^{r,J}]^\top \in \{0,1\}^J$ based on local importance indicator $Q_{k,l}^r$. Given the sparse ratio $s_k^r$, client $k$ calculates an importance threshold $\tau_{k,l}^r$ as $(1-s_k^r)$-quantile of $Q_{k,l}^r$ such that the sparse pattern is formulated by
\begin{equation}
	P_{k,l}^r = \gamma(Q_{k, l}^r - \tau_{k,l}^r I),
\end{equation}
where $\gamma(\cdot)$ is a step function and $I$ denotes a unit vector with the same shape as $Q_{k, l}^r$.
If $q_{k,l}^{r,j} < \tau_{k,l}^r$,  the $j$-th element in $P_{k,l}^r$ is set to $\beta_{k,l}^{r,j} = 0$ and the connections corresponding to the $j$-th model units are masked. In this way, the local mask can be derived in each local iteration $l$ by
\begin{equation}
		\label{local_mask_customize}
	 	m^r_{k,l} = \mathcal{M}(P_{k,l}^r\; | \; \omega_{k,l}^r, s^r_k) 
\end{equation}
where $\omega_{k,l}^r$ denotes local parameters in current iteration. We expect to tailor a personalized model with a learnable sparse pattern to fit local data. Hence, a personalized importance indicator should be learned to customize the sparse pattern.

\paratitle{Importance-Associated Regularization Loss.} We design an importance-associated regularization loss, which consists of three terms. The first one is the task-specific optimization function (\emph{e.g.}, cross-entropy function) between the sparse model prediction $\hat{y}$ and the data label $y$ for any $(x,y)\in D^k$:
\begin{equation}
	\label{first_term}
	\mathcal{L}_{tr}^k = \ell\Big(\hat{y},y | \, \omega_{k,l}^r \odot \mathcal{M}\big(\gamma(Q_{k,l}^r - \tau_{k,l}^r I)\; |\;\omega_{k,l}^r, s_k^r\big) \Big).
\end{equation}
By (\ref{first_term}), we establish the correlation between the unit significance $Q_{k,l}^r$ and local data $D^k$, which enables the importance indicator to be learned through local data mining. 

The second one is a local parameter regularization term
\begin{equation}
	\label{parameter_regularization}
	\mathcal{L}_{pr}^k = \big\|\omega_{k,l}^r - \omega^r \big\|^2.
\end{equation}
 It limits local updates not to deviate too much and is commonly used in prior works \cite{ditto_PFL_2021, FedTrip_2023}. The third one is an importance regularization term to prevent excessive shifts and over-sharpening of the importance indicator, formulated as:
 \begin{equation}
 	\label{importance_regularization}
 	\mathcal{L}_{ir}^k = \big\|Q_{k,l}^r - \sigma(|\omega_{k,l}^r|_{\mathrm{J}})\big\|^2.
 \end{equation}
Here, $|\omega_{k,l}^r|_{\mathrm{J}}$ is a sum vector of parameter absolute values, with $J$ dimensions. The $j$-th element of $|\omega_{k,l}^r|_{\mathrm{J}}$ is the sum of the absolute value of parameters corresponding to the $j$-th unit. $\mathcal{L}_{pr}^k$ has constrained the update of local parameters so that their magnitude sum $|\omega_{k,l}^r|_{\mathrm{J}}$ would not change sharply. $\sigma(\cdot)$ is a sigmoid function, which further smooths $|\omega_{k,l}^r|_{\mathrm{J}}$ and limits it to $[0,1)$.  Hence, (\ref{importance_regularization}) avoids $Q^r_{k,l}$ from being too biased or over-sharpened during local training. Combining the three terms, the importance-associated regularization loss is expressed as
\begin{equation}
\label{our_loss}
\mathcal{L}_k (\omega_{k,l}^r,\; Q_{k,l}^r \, | \, \omega^r, D^k)= \mathcal{L}_{tr}^k  + \mu * \mathcal{L}_{pr}^k + \lambda * \mathcal{L}_{ir}^k.
\end{equation} 
In this design, we integrate the importance indicator into loss to make unit significance learnable, which assists in updating the importance-based sparse pattern with minimal heuristics.

\paratitle{Client-side Update.} With global parameters $\omega^{r}$ and sparse ratio $s_k^r$, client $k$ measures its computation capability $z_k$ and restricts the sparse model under the carrying capacity by $s_k^r \leq z_k$. If the server-determined sparse ratio exceeds local capability, the client directly resets $s_k^r = z_k$. Consistent with common works \cite{HeteroFL_2021, Fjord_2021}, we perform layer-wise sparsification and adopt the same sparse ratio $s_k^r$ for each layer. The sparsification is induced into the local initial model via $\omega_{k,0}^r \odot \mathcal{M}\big(\gamma(Q_{k, 0}^r - \tau_{k, 0}^r I)|\;\omega_{k,l}^r, s_k^r \big)$, where global parameters $\omega^r$ are imported into the local model as $\omega_{k,0}^r$. The iterative updates of local parameters and importance-based sparse patterns (lines 18-22 in Algorithm \ref{FedLPS}) are described as follows. For each local iteration $l=\{0,\ldots, E-1\}$, client $k$ constructs the local mask $m_{k,l}^r$ with the importance indicator by (\ref{local_mask_customize}) and induces sparseness into the local model by $\omega_{k,l}^r \odot m_{k,l}^r$. In back-propagation, local sparse parameters are updated via
\begin{align}
\label{parameter_update}
\omega_{k,l+1}^r = \omega_{k,l}^r  & \; - \;  \textbf{SGD}\Big(\mathcal{L}_k(\omega_{k,l}^r,\; Q_{k,l}^r \, | \, \omega^r, d^k_l), \; \\ & \quad \omega_{k,l}^r \odot \mathcal{M}\big(\gamma(Q_{k, l}^r - \tau_{k, l}^r I) |\;\omega_{k,l}^r, s_k^r \big), \eta_r \Big), \notag
\end{align}
and the update of the unit-wise importance indicator follows
\begin{equation}
\label{importance_update}
Q_{k,l+1}^r = Q_{k,l}^r  -  \textbf{SGD}\Big(\mathcal{L}_k(\omega_{k,l}^r, Q_{k, l}^r | \omega^r, d^k_l), Q_{k,l}^r, \eta_r \Big).
\end{equation}
Here, $\textbf{SGD}(\mathcal{L}, Q, \eta)$ denotes the gradient calculation of function $\mathcal{L}$ on $Q$ with a learning rate $\eta$, and $d^k_l$ is the training data subset used in the $l$-th local iteration. The joint optimization of sparse parameters and importance indicators (that derive sparse patterns) enables the learnable sparse training on local specific data, further customizing a sparse model for each client. 

\paratitle{Parameter Upload and Global Aggregation.} After $E$ iterations, client $k$ locally stores its personalized model $\Tilde{\omega}_{k,E}^r = \omega_{k,E}^r \odot m_{k, E}^r$ with $m_{k,E}^r = \mathcal{M}(\gamma(Q_{k, E}^r - \tau_{k, E}^r I) | \omega_{k,E}^r, s_k^r)$ (as lines 23-24 in Algorithm \ref{FedLPS}). The nonzero parameters of residual gradients 
\begin{equation}
	\hat{\omega}_{k,E}^r = (\omega^{r} - \omega_{k,E}^r) \odot m_{k, E}^r
\end{equation}
are uploaded to the server. Subsequently, the server performs global aggregation (line 8 in Algorithm \ref{FedLPS}) via
\begin{equation}
\label{aggreagtion}
\omega^{r+1} = \frac{\sum_{k\in\mathcal{C}_r}{|D_k|\; (\omega^{r} - \hat{\omega}_{k,E}^r})}{\sum_{k\in\mathcal{C}_r}{|D_k|}}.
\end{equation}
Note that $\hat{\omega}_{k,E}^r$ denotes the sparse local update, while $\omega^{r}$ represents dense global parameters, which implies that $(\omega^{r} - \hat{\omega}_{k,E}^r)$ is relatively dense. Furthermore, the local sparse pattern of $\hat{\omega}_{k,E}^r$ is unique for each client $k$. 
Thus, the aggregation can provide a relatively dense update for global parameters.

\subsection{Online Decision of Sparse Ratio}

In complicated non-IID and system-heterogeneous EDM, the decision of sparse ratios has to consider two aspects. \emph{First, the sparse ratio intuitively determines the model scale and training costs.}
The lower sparse ratio contributes to a smaller submodel and fewer costs, which can match lower-tier clients.
\emph{Secondly, the sparse ratio is closely correlated with model accuracy over local data} in non-IID settings, where the lower sparse ratio is more likely to deteriorate accuracy. 
However, most existing works manually set rigid rules, failing to jointly optimize the two aspects. They either directly specify a fixed adjustment rate \cite{LotteryFL_2021, Hermes_2021} or simply control sparse ratios according to device capabilities \cite{Fjord_2021, HeteroFL_2021}, leading to training delays and performance bottlenecks.
Thus, it is essential to adaptively determine sparse ratios by learning the additive effect of heterogeneous capabilities and non-IID data for the trade-off between capability adaption and accuracy guarantee.

In the whole FL process, the sparse ratio decision for each client can be regarded as a sequential decision problem, which perfectly matches the modeling of the Multi-Armed Bandit (MAB) problem. The server is viewed as a bandit and arms are feasible sparse ratios. Within limited rounds, the server decides sparse ratios for selected clients (\emph{i.e.}, the bandit chooses arms) in each round. Different sparse ratios bring distinct local costs and training accuracy for each client, which reflect the rewards of arms in our case. We aim at obtaining a sparse ratio sequence for each client to achieve accuracy guarantee with as little training overhead as possible. Thus, we model the ratio decision for each client as a MAB problem, and the server creates $K$ agents to address the MAB problems of $K$ clients. For each agent, the arm space $[0,1)$ of sparse ratios is infinite. 
	
Although several UCB-extended methods \cite{EUCBV_2018, FedMP_2022} are explored for MAB problems, they either only work with discrete arms without involving the transform of infinite arm space or unilaterally target time saving without considering non-IID data, bringing performance sacrifice and even hindering convergence. In this work, we develop Prompt Upper Confidence Bound Variance (P-UCBV), in which a novel reward function is designed to learn the additive effect of resource restrictions and non-IID data on sparse ratios for the trade-off between resource self-adaption and accuracy improvement. Moreover, we introduce the accuracy-dominated arm elimination, where the sparse ratios that sharply deteriorate accuracy on local data are promptly removed from the feasible arm space to avoid severe accuracy fluctuations and improve decision efficiency.

\paratitle{Reward Function.} Considering the superimposed effect of distinct resource restrictions and non-IID data, both time costs and accuracy changes are involved in the reward function. First, we build a local cost formula involving computation and communication overhead. The local cost of client $k$ in round $r$ is denoted as $T^r_k$. Let Floating Point Operations (FLOPs) characterize computation overhead and $\widehat{F}^r_k$ denote FLOPs of client $k$ in round $r$. The communication cost is represented by the transmitted parameter size $\widehat{B}^r_k$. The local cost $T^r_k$ under client-side resource configurations is calculated by
\begin{equation}
\label{cost}
T^r_k = \widehat{F}^r_k / F_k^r + \alpha \widehat{B}^r_k / B_k^r,
\end{equation} 
where $F_k^r$ and $B_k^r$ are maximum capacities of locally available computation and bandwidth. The cost calculation of (\ref{cost}) is implemented on the client, meaning that privacy-sensitive configuration information (\emph{e.g.}, computing power $F_k^r$) of clients will not be leaked. The reward of sparse ratio $s_k^r$ is defined as
\begin{equation}
\label{reward}
G(s_k^r) = \big(U(a_k^r) -U(a_k^{r-1})\big) / T_k^r.
\end{equation}
The utility function $U(\cdot)$ \cite{multi-agent_RL_2022} is used to moderately transform the accuracy, which accounts for marginal accuracy gains near the end of FedLPS.

\begin{algorithm}[tb]
\caption{P-UCBV in round $r$ for client $k \in\mathcal{C}_r$}
\label{EUCBV}
\small
\textbf{Initialize}: partition set $\textbf{S}_{k,0}$;
client selection fraction $\epsilon$; differential accuracy threshold $\Delta$; $\varepsilon_0 : = 1$; $\xi=R/(K\cdot \epsilon);\; \psi=\xi / I_0^2; \;$ \\
\textbf{Input}: local cost $T_k^r$; average training accuracy $a_k^r$;\\ 
\begin{algorithmic}[1] 
	\STATE $S_k^u \gets$ the partition where $s_k^r$ resides
	\STATE Split $S_k^u$ with $s_k^r$ into $S_k^{u'}, S_k^{u''}$ for buliding $\textbf{S}_{k,r+1}$
	\IF {$a_k^r - a_k^{r-1} < \Delta$}
	\STATE Remove $S_k^{u'}$ from $\textbf{S}_{k,r+1}$  \qquad \emph{// Arm elimination}
	\ENDIF
	\STATE $\varepsilon_{r+1} \gets \varepsilon_r /2;\;\; I_{r+1} \gets |\textbf{S}_{k,r+1}|;$
	\STATE $ \psi\gets\xi / I_{r+1}^2$
	\STATE Calculate reward by (\ref{reward}), which is added into the reward lists of partitions $S_k^{u'}$ (if it exists) and $ S_k^{u''}$
	\STATE Count the average reward $\{\bar{g}_k^i \,| \, i=1, \ldots, I_{r+1}\}$ and variance $\{\bar{v}_k^i \,| \, i=1, \ldots, I_{r+1}\}$ of all partitions
	\STATE Select partition $S_k^e=\arg\max_{S_k^i\in\textbf{S}_{k,r+1}} \mathcal{U}(S_k^i)$ 
	\STATE Sample $s_k^{r+1}$ from $S_k^e$
	\STATE \textbf{return} $s_k^{r+1}$
\end{algorithmic}
\end{algorithm}

\paratitle{P-UCBV Algorithm.} The detailed process of P-UCBV is described in Algorithm \ref{EUCBV}. For each agent, the infinite ratio space is divided based on the decision tree \cite{FedMP_2022}. Initially, agent $k$ holds $I_0$ ratio partitions $\textbf{S}_{k,0}=\{S^1_{k}, \ldots, S^{I_0}_{k}\}$ with $\bigcup_{i=1}^{I_0}S^i_{k}=[0,1)$ and randomly chooses a partition to sample the initial sparse ratio $s_k^0$ from the selected partition. We evaluate the initial global model on local data to obtain the original accuracy set $\{a_k^{-1}|k=1,\ldots, K\}$. In each round $r$, the server splits the last selected partition $S_k^u$ (\emph{i.e.}, $s_k^r \in S_k^u$) into two partitions $S_k^{u'}$ and $S_k^{u''}$, where $s_k^r$ is the split point (line 2). If $s_k^r$ causes a larger sacrifice in accuracy, the accuracy-dominated prompt arm-elimination operation is activated, and the partition $S_k^{u'}$ is promptly removed to build a new partition set $\textbf{S}_{k,r+1}$ (lines 3-5). Subsequently, the reward of the selected sparse ratio $s_k^r$ is calculated by Equation (\ref{reward}) (line 8). We count the average reward of the $i$-th arm (\emph{i.e.}, $S_k^i$) as
\begin{equation}
	\bar{g}^i_k=\frac{1}{h_k^i}\sum_{\ell=1}^{h_k^i} G^i_{k,\ell},
\end{equation}
where $G^i_{k,\ell}$ is the reward feedback when $S_k^i$ is chosen for the $\ell$-th time, and $h_k^i$ is pulled times of $S_k^i$. The variance is calculated by 
	$\bar{v}^i_k = \frac{1}{h_k^i}\sum_{\ell=1}^{h_k^i} (G^i_{k,\ell} - \bar{g}^i_k)^2$.
With $\bar{g}^i_k$ and $\bar{v}^i_k$, we compute the UCBV value of each partition $S_k^i \in \textbf{S}_{k,r+1}$ by
\begin{equation}
	\mathcal{U} (S_k^i) = \bar{g}^i_k+\sqrt{\frac{\rho (\bar{v}^i_k+2) \log(\xi\psi\varepsilon_{r+1})}{4(h_k^i+1)}},
\end{equation}
where $\xi=R/(K\cdot \epsilon)$ and $\rho$ is a preset constant. $\varepsilon_{r+1}$ and $\psi$ are updated in P-UCBV as lines 6-7.
Then, the optimal partition is selected to sample a new sparse ratio (lines 10-11). 



\section{Further Analysis}

\subsection{Cost Analysis}
\paratitle{Local Computation Cost.} We analyze the local computation costs of FedLPS, where the updates of sparse parameters and importance indicators require computing support. The nonzero parameters of the sparse model are locally trained, significantly alleviating the local computation burden compared to the original dense model. On the other hand, the computation cost for the importance indicator updating is much less than that of model parameters, which can be ignored. Because the size of a unit-wise importance indicator is much smaller than the size of model parameters. Considering a model consisting of three fully-connected layers with 1024 neurons, the FLOPs of updating the local model are 15.36 $\times$ $10^5$ in an iteration, following the FLOP calculation in \cite{DisPFL_2022}. While the FLOPs of updating importance indicators are 750, which is much less than 15.36 $\times$ $10^5$ and even can be ignored. Besides, sparse ratios are decided on the server with adequate computing power, without involving local computation consumption. Hence, FedLPS effectively reduces local computation costs. 

\paratitle{Communication Cost.} For uplink communication, only nonzero parameters of local sparse models and tiny unit-wise binary patterns are uploaded, which can significantly mitigate uplink overhead. In terms of downlink communication, the server delivers relatively dense global parameters and selected sparse ratios (much less than global parameters and even can be ignored) to clients. In this way, FedLPS has a similar downlink overhead as conventional FL frameworks (\emph{e.g.}, FedAvg).

\paratitle{Global Cost.} FedLPS adopts synchronous aggregation such that the global time cost is determined by the slowest client. In round $r$, the global time cost is modeled by 
\begin{equation}
T^r = max_{k \in \mathcal{C}_r} T_k^r,
\end{equation} 
where $T_k^r$ is calculated by (\ref{cost}). 
Based on the above analysis, FedLPS can mitigate client-side computation burden (\emph{i.e.}, $\hat{F}_k^r$) and uplink communication volume (\emph{i.e.}, $\hat{B}_k^r$) such that local time costs $T_k^r, k \in \mathcal{C}_r$ are also reduced. When $T_k^r$ decreases, the corresponding global cost $T^r$ will be reduced. Our experiments have demonstrated that FedLPS reduces the total time cost to accelerate training, as shown in Section \ref{experiment}. 

\subsection{Privacy Analysis}

FedLPS transmits residual model parameters without sharing privacy-sensitive raw data, similar to conventional FL frameworks, which mitigates data privacy concerns in edge networks. For transmitting local cost and training accuracy, previous work \cite{FedMP_2022} has verified its reliability and practicality. As mentioned in \cite{FedMP_2022}, the server needs to be online aware of the different capabilities of clients in heterogeneous FL settings, while the decision-making strategy based on transmitted local time costs and training evaluation results can avoid the direct leakage of privacy-sensitive computing power information, practical in heterogeneous edge computing.

Furthermore, FedLPS is orthogonal to existing FL privacy-preserving techniques (\emph{e.g.}, differential privacy \cite{Differential_Privacy} and homomorphic encryption \cite{Homomorphic_Encryption_FL_2020}), which can be directly applied with FedLPS. For instance, we could add noise to transmitted parameters via differential privacy, and encrypt local costs and training accuracy to further guarantee privacy security.

\subsection{Convergence Analysis}

In this section, we provide a formal theoretical analysis to guarantee the convergence of FedLPS under Stochastic Gradient Descent (SGD) optimization with the constraints of learning rates. Let $\Tilde{\omega}^r_{k,l} = \omega^r_{k,l} \odot m^r_{k,l}$ with $m^r_{k,l} = \mathcal{M}\big(\gamma(Q_{k, l}^r - \tau_{k, l}^r I) | \,\omega_{k,l}^r, s_k^r \big)$ and $ \nabla_{\Tilde{\omega}} \mathcal{L}_k(\Tilde{\omega}^r_{k,l}, Q^r_{k,l}; \xi_l) = 1/ \eta_r \cdot \textbf{SGD}(\mathcal{L}_k (\omega_{k,l}^r,\, Q_{k,l}^r\,|\,\omega^r,\, d^k_l),\, \omega_{k,l}^r \odot \mathcal{M}\big(\gamma(Q_{k, l}^r - \tau_{k, l}^r I) | \omega_{k,l}^r, s_k^r \big), \eta_r )$  denotes the estimate of $\nabla_{\Tilde{\omega}}F_k(\Tilde{\omega}^r_{k,l})$, where $\xi_l$ is a random variable. Besides, we define the optimal local mask of client $k$ as $m^*_k$.

To facilitate analysis, we make the following assumptions.
\begin{assumption} [\textbf{$L$-Lipschitz Smoothness}]
	\label{smoothness}
	 For each client $k\in\{1,\ldots, K\}$, the function $F_k$ is smooth such that
	 \begin{equation*}
	 	\| \nabla_{\omega} F_k (\omega)- \nabla_{\omega} F_k (\omega')\| \leq L \|\omega- \omega' \|.
	 \end{equation*}
\end{assumption}
\begin{assumption} [\textbf{Bounded Gradient Estimator Bias}] 
	\label{estimated_bias}
	The stochastic gradient estimation of each client $k\in\{1,\ldots, K\}$ satisfies $\sigma^{2}$-bounded bias with
	\begin{equation*}
		\mathbb{E}  \big\| m^r_{k,l} \odot \nabla_{\Tilde{\omega}} \mathcal{L}_k(\Tilde{\omega}^r_{k,l}, Q^r_{k,l}; \xi_l) - m^r_{k,l} \odot \nabla_{\Tilde{\omega}} F_k(\Tilde{\omega}^r_{k,l}) \big\|^2 \leq \sigma^2
	\end{equation*}
	for all $r\in\{0,\ldots,R-1\}$ and $l\in\{0,\ldots,E-1\}$.
\end{assumption}
\begin{assumption} [\textbf{Bounded Local Sparse Gradient}] 
	\label{global_variance}
	For each client $k\in\{1,\ldots, K\}$, the expected squared norm of local sparse gradients is bounded by $H^2$, \emph{i.e.},
	\begin{equation*}
		\mathbb{E}\big\| m^r_{k,l} \odot \nabla_{\Tilde{\omega}} F_k( \Tilde{\omega}^r_{k,l}) \big\|^2 \leq H^2.
	\end{equation*}
\end{assumption}
\begin{assumption} [\textbf{Bounded Sparse Gradient Distance}]
	\label{similarity}
	For each client $k\in\{1,\ldots, K\}$, the distance between local sparse gradients with the optimal mask and average sparse gradients of all clients are bounded by a constant $B$ with
	\begin{equation*}
		\Big\| m_k^* \odot \nabla_{\Tilde{\omega}} F_k(\Tilde{\omega}_k) - \frac{1}{K} \sum_{i=1}^{K} m_i^* \odot \nabla_{\Tilde{\omega}} F_i(\Tilde{\omega}_i) \Big\| \leq B,
	\end{equation*}
	where $\Tilde{\omega}_k = m_k^* \odot \omega_k$ for any $\omega_k \in \Omega$.
\end{assumption}
Assumptions 1-3 are common in many FL convergence studies \cite{fedavg_2017, non-IID_problem_2023, PruneFL_2022}. Assumption \ref{similarity} bounds the sparse gradient difference between the local loss and global average loss, which is widely used to characterize client diversities \cite{chen2023enhancing, fedcm_2021}. 
Based on these assumptions, we analyze the convergence of FedLPS. In round $r$, a uniform learning rate $\eta_r$ is adapted for all selected clients and the local parameter update in the $l$-th iteration can be rewritten as $\omega^r_{k,l+1} = \omega^r_{k,l} - \eta_r \cdot m^r_{k,l} \odot \nabla_{\Tilde{\omega}} \mathcal{L}_k(\Tilde{\omega}^r_{k,l}, Q^r_{k,l}; \xi_l)$. Firstly, we give the upper bound of the gap between global and local parameters, as shown in Lemma \ref{Drift}.

\begin{lemma}
	\label{Drift}
	Let Assumptions 1-4 hold. For any $r \in \{0,\dots,R-1\}$, $	l \in \{0,\dots,E-1\}$, it follows
	\begin{equation*}
			\frac{1}{K} \sum_{k=1}^{K} \, \mathbb{E}  \big\| \omega^r_{k,l} -  \omega^{r+1} \big\|^2 \leq  5E\eta_r^2  \big ( \sigma^2 + 6E B^2 + 18E H^2 \big) 
	\end{equation*} 
	with the learning rate $\eta_r \leq \sqrt{\frac{1}{24ERV_rL^2}}$, where $V_r = \max_{k\in [K], l\in[E]} \Big\{ \frac{ \|m^r_{k,E} \odot m^r_{k,l} \odot \nabla_{\Tilde{\omega}} F_k(\Tilde{\omega}^r_{k,l}) - m_k^* \odot \nabla_{\Tilde{\omega}} F_k(\Tilde{\omega}_{k,r}) \|^2 } { \| \nabla_{\omega} F_k({\omega}^r_{k,l}) - \nabla_{\omega} F_k({\omega}^{r+1}) \|^2} \Big\}$ with $\Tilde{\omega}_{k,r} = m_k^* \odot \omega_{k,0}^r$ and $\Tilde{\omega}^r_{k,l} = m^r_{k,l} \odot \omega_{k,l}^r $.
\end{lemma}
\begin{proof} To prove the above result, we define $q^r_{k,l} (\Tilde{\omega}^r_{k,l}) = \nabla_{\Tilde{\omega}} \mathcal{L}_k(\Tilde{\omega}^r_{k,l}, Q^r_{k,l}; \xi_l)$ and specify that
 \begin{itemize}
 	\item $A_1 = m^r_{k,l-1} \odot \ q_{k,l-1}^r(\Tilde{\omega}^r_{k,l-1})  - m^r_{k,l-1} \odot \nabla_{\Tilde{\omega}} F_k(\Tilde{\omega}^r_{k,l-1})$
 	\item $A_2 = m^r_{k,l-1} \odot \nabla_{\Tilde{\omega}} F_k(\Tilde{\omega}^r_{k,l-1}) - m_k^* \odot \nabla_{\Tilde{\omega}} F_k(\Tilde{\omega}_{k,r})$
 	\item $A_3 = m_k^* \odot \nabla_{\Tilde{\omega}} F_k(\Tilde{\omega}_{k,r})  - \frac{1}{K} \sum_{i=1}^{K} m_i^* \odot \nabla_{\Tilde{\omega}} F_i( \Tilde{\omega}_{i,r})$
 	\item  $A_4 = \frac{1}{K} \sum_{k=1}^{K} m_k^* \odot \nabla_{\Tilde{\omega}} F_k( \Tilde{\omega}_{k,r})$
 \end{itemize}
In this way, there is
	\begin{align}
		& \mathbb{E} \| \omega^r_{k,l} -  \omega^{r+1} \|^2 \\
		& = \mathbb{E} \big\| \omega^r_{k,l-1} - \omega^{r+1} - \eta_r m^r_{k,l-1} \odot  q^r_{k,l-1} (\Tilde{\omega}^r_{k,l-1}) \big\|^2 \notag \\
		& \leq \Big(1 + \frac{1}{2E-1}\Big) \mathbb{E} \big\| \omega^r_{k,l-1} - \omega^{r+1}\big\|^2 +  \eta_r^2 \, \mathbb{E} \big\| A_1 \big\|^2 \notag \\
		&\quad + 6E\eta_r^2 \, \mathbb{E} \big\| A_2 \big\|^2 + 6E\eta_r^2 \, \mathbb{E} \big\| A_3 \big\|^2  + 6E\eta_r^2 \, \mathbb{E} \big\| A_4 \big\|^2. \notag
	\end{align}
\textbf{Bounding $\bm{A_1}$} by Assumption  \ref{estimated_bias}, it holds
\begin{equation*}
	\eta_r^2 \, \mathbb{E} \|  m^r_{k,l-1} \odot \big(q_{k,l-1}^r(\Tilde{\omega}^r_{k,l-1})  -  \nabla_{\Tilde{\omega}} F_k(\Tilde{\omega}^r_{k,l-1}) \big) \|^2 \leq \eta_r^2 \sigma^2.
\end{equation*}
\textbf{In terms of $\bm{A_2}$}, we observe that
\begin{align}
	& 6E \eta_r^2 \mathbb{E} \big\| m^r_{k,l-1} \odot \nabla_{\Tilde{\omega}} F_k(\Tilde{\omega}^r_{k,l-1}) - m_k^* \odot \nabla_{\Tilde{\omega}} F_k(\Tilde{\omega}_{k,r}) \big\|^2 \notag \\
	& \leq 12EV_r\eta_r^2\mathbb{E} \| \nabla_{\omega} F_k({\omega}^r_{k,l-1}) - \nabla_{\omega} F_k({\omega}^{r+1}) \big\|^2 + 12EH^2\eta_r^2 \notag \\
	& \leq 12EV_rL^2\eta_r^2 \mathbb{E} \| {\omega}^r_{k,l-1} - \omega^{r+1} \|^2 + 12EH^2\eta_r^2,
\end{align}
where the last inequality holds by Assumption \ref{smoothness}. \\
\textbf{Bounding $\bm{A_3}$} with Assumption \ref{similarity}, there exists 
\begin{align}
	6E\eta_r^2 & \, \mathbb{E} \Big\| m_k^* \odot \nabla_{\Tilde{\omega}} F_k(\Tilde{\omega}_{k,r})  - \frac{1}{K} \sum_{i=1}^{K} m_i^* \odot \nabla_{\Tilde{\omega}} F_i( \Tilde{\omega}_{i,r}) \Big\|^2  \notag \\ 
	& \leq 6E\eta_r^2 B^2.
\end{align}
\textbf{For $\bm{A_4}$}, due to Assumption \ref{global_variance} and Lemma 1 in \cite{FedSpa_2022}, it satisfies 
\begin{equation}
	6E\eta_r^2\, \mathbb{E} \Big\| \frac{1}{K} \sum_{k=1}^{K} m_k^* \odot \nabla_{\Tilde{\omega}} F_k( \Tilde{\omega}_{k,r}) \Big\|^2 \leq 6E\eta_r^2 H^2.
\end{equation} \textbf{Combining the four terms} together, we obtain the following:
\begin{align}
	& \mathbb{E} \| \omega^r_{k,l} -  \omega^{r+1} \|^2 \\
	& \leq \Big(1 + \frac{1}{2E-1}\Big) \, \mathbb{E} \| \omega^r_{k,l-1} - \omega^{r+1}\|^2 + \eta_r^2 \sigma^2 \notag \\
	&  \quad + 6E\eta_r^2(B^2+ 3H^2)  + 12EV_rL^2\eta_r^2\, \mathbb{E} \| {\omega}^r_{k,l-1} - \omega^{r+1} \|^2. \notag
\end{align}
Considering \textbf{the average of all clients}, it follows
\begin{align}
	\frac{1}{K} & \sum_{k=1}^{K} \, \mathbb{E} \| \omega^r_{k,l} -  \omega^{r+1} \|^2 \notag \\
	& \leq \Big(1 + \frac{1}{E-1}\Big)\frac{1}{K} \sum_{k=1}^{K} \mathbb{E} \big\| \omega^r_{k,l-1} - \omega^{r+1} \big\|^2 \notag \\
	& \quad + \eta_r^2 \sigma^2 + 6E\eta_r^2(B^2+ 3H^2),
\end{align}
where the first inequality holds because $\eta_r \leq \sqrt{\frac{1}{24ERV_rL^2}}$ such that $\frac{1}{2E-1} + 12EV_rL^2\eta_r^2 \leq \frac{1}{2E-1} + \frac{1}{2R} \leq \frac{1}{E-1}$ with $R \geq E-1$ and $E>1$. \textbf{Expanding the recursion}, we get
\begin{equation*}
	\frac{1}{K} \sum_{k=1}^{K} \, \mathbb{E} \| \omega^r_{k,l} -  \omega^{r+1} \|^2 \leq 5E\eta_r^2 (\sigma^2 +6EB^2+ 18EH^2),
\end{equation*}
concluding the proof of Lemma \ref{Drift}.
\end{proof}

Based on Lemma \ref{Drift}, we can provide the convergence result of FedLPS in Theorem \ref{convergence}.

\begin{theorem}
	 \label{convergence}
	Let Assumptions 1-4 hold. Choose $\phi = 4\sqrt{6}L\max_{r\in[R]}\sqrt{V_r}$, $\varphi = \max_{r\in[R]} \sqrt{\frac{E}{6V_r}}$, and the learning rate $\eta_r \leq \sqrt{\frac{1}{24ERV_rL^2}}$. Then for FedLPS, it follows
	\begin{align}
		\frac{1}{R} & \sum_{r=0}^{R-1} \mathbb{E} \Big\| \frac{1}{K}\sum_{k=1}^{K} \nabla_{\Tilde{\omega}} F_k(\Tilde{\omega}_{k,r})\Big\|^2 \notag \\
		& \leq \frac{\phi}{\sqrt{ER}}\big(f_0 - f^*\big) + \frac{\varphi}{\sqrt{R}}\Big (2H^2+\frac{\sigma^2}{KE} \Big) \notag \\
		& \quad + \frac{1}{R}\Big(\frac{5}{24}+\frac{5\varphi}{12\sqrt{R}}\Big)(\sigma^2+6EB^2+18EH^2),
	\end{align}
	where $f_0 = \frac{1}{K}\sum_{k=1}^{K}F_k(\Tilde{\omega}_{k,0})$ and $f^* = \frac{1}{K}\sum_{k=1}^{K}F_k(\Tilde{\omega}_k^*)$ with the optimal local sparse parameters $\Tilde{\omega}_k^*$.
\end{theorem}
\begin{proof}
	To verify the above theorem, we define $f (\omega^r)= \frac{1}{K}\sum_{k=1}^{K}F_k(m_k^* \odot \omega^r) =  \frac{1}{K}\sum_{k=1}^{K}F_k(\Tilde{\omega}_{k,r}) $ since $\omega^r_{k,0} = \omega^r $
	and $\nabla f(\omega^r) = \frac{1}{K}\sum_{k=1}^{K} \nabla_{\Tilde{\omega}} F_k(\Tilde{\omega}_{k,r}) $. 
	By Assumption \ref{smoothness} and Lemma 5 in \cite{FedSpa_2022}, there is 
	\begin{equation}
		\small
		\label{total}
		f(\omega^{r+1}) \leq f(\omega^r) - \langle \nabla f(\omega^r), \omega^{r+1} - \omega^r \rangle + \frac{L}{2} \| \omega^{r+1}-\omega^r \|^2.
	\end{equation}
	 \textbf{For the second term} on the right of (\ref{total}), we derive that
\begin{align}
	\label{T1}
		- & \mathbb{E} \big[ \langle \nabla f(\omega^r), \omega^{r+1} - \omega^r \rangle \big] \\ 
		& \leq \frac{\eta_r}{2E} \Big( \mathbb{E} \Big\|  \frac{1}{K} \sum_{k=1}^K \sum_{l=0}^{E-1}  \big [ m^r_{k,E} \odot m^r_{k,l} \odot \nabla_{\Tilde{\omega}}F_k (\Tilde{\omega}^r_{k,l}) \notag \\ 
		&\qquad \qquad \quad - E \nabla f(\omega^r) \big ] \Big\|^2 - E^2\mathbb{E} \| \nabla f(\omega^r) \|^2 \Big) \notag \\ 
		& \leq \frac{\eta_rV_rL^2}{2K} \sum_{k=1}^K \sum_{l=0}^{E-1} \mathbb{E} \big\| \omega_{k,l}^r - \omega^{r+1} \big\|^2 - \frac{\eta_rE}{2}\mathbb{E} \| \nabla f(\omega^r) \|^2. \notag
\end{align}
	 \textbf{For the third term} on the right of (\ref{total}), it holds:
	\begin{align}
		\label{T2}
		\frac{L}{2} & \, \mathbb{E} \| \omega^{r+1}-\omega^r \|^2 \notag \\ 
		& \leq L\eta_r^2\mathbb{E} \Big\| \frac{1}{K} \sum_{k=1}^K m^r_{k,E} \odot \sum_{l=0}^{E-1} m^r_{k,l} \odot \nabla_{\Tilde{\omega}}F_k (\Tilde{\omega}^r_{k,l}) \Big\|^2 \notag \\
		& \qquad+ \frac{LE\eta_r^2\sigma^2}{K} \notag \\
		& \leq \frac{2LEV_r\eta_r^2}{K} \sum_{l=0}^{E-1} \sum_{k=1}^K \| \nabla_{\omega} F_k (\omega^r_{k,l}) -  \nabla_{\omega} F_k (\omega^{r+1})\|^2 \notag \\
		& \quad + 2 LE^2\eta_r^2\,\mathbb{E} \Big\| \frac{1}{K} \sum_{k=1}^{K} m_k^* \odot \nabla_{\Tilde{\omega}} F_k( \Tilde{\omega}_{k,r}) \Big\|^2  + \frac{LE\eta_r^2\sigma^2}{K} \notag \\
		& \leq \frac{2L^3EV_r\eta_r^2}{K} \sum_{l=0}^{E-1} \sum_{k=1}^K \mathbb{E} \|\omega^r_{k,l} - \omega^{r+1}\|^2 + 2LE^2H^2\eta_r^2 \notag \\
		& \qquad + \frac{LE\eta_r^2\sigma^2}{K}.
	\end{align}
	\textbf{Combining (\ref{T1}) and (\ref{T2})}, we can obtain the following:
	\begin{align}
		\label{before_remove}
		\mathbb{E} [f(\omega^{r+1})] \leq & \mathbb{E} [f(\omega^{r})] -  \frac{\eta_rE}{2}\mathbb{E} \| \nabla f(\omega^r) \|^2 \notag \\
		& +  \frac{5E^2L^2V_r\eta_r^3}{2} (\sigma^2+6EB^2+18EH^2) \notag \\
		& + 10E^3L^3V_r\eta_r^4 (\sigma^2+6EB^2+18EH^2) \notag \\ 
		& + 2E^2LH^2\eta_r^2+ \frac{LE\eta_r^2\sigma^2}{K}.
	\end{align}
	Furthermore, we \textbf{transform inequality (\ref{before_remove})} into
	\begin{align}
		& \mathbb{E} \| \nabla f(\omega^{r})\|^2 \notag \\ 
		& \leq  \frac{2\{\mathbb{E} [f(\omega^{r})] - \mathbb{E} [f(\omega^{r+1})]\}}{\eta_rE}  + 2EL\eta_r\Big (2H^2+\frac{\sigma^2}{KE} \Big) \notag \\
		& \quad + 5EL^2V_r\eta_r^2(1+4EL\eta_r)(\sigma^2+6EB^2+18EH^2).
	\end{align}
	By $\eta_r \leq \sqrt{\frac{1}{24ERV_rL^2}}$, there is
		\begin{align}
		\mathbb{E} \| \nabla f(\omega^{r})\|^2 \leq &  \frac{4\sqrt{6RV_r}L \{ \mathbb{E} [f(\omega^{r})] - \mathbb{E} [f(\omega^{r+1})]\}}{\sqrt{E}} \notag \\
		& + \frac{5}{24R}(1+\frac{2\sqrt{E}}{\sqrt{6RV_r}})(\sigma^2+6EB^2+18EH^2) \notag \\
		& + \frac{\sqrt{E}}{\sqrt{6RV_r}}\Big (2H^2+\frac{\sigma^2}{KE} \Big).
	\end{align}
	Let $\phi = 4\sqrt{6}L\max_{r\in[R]}\sqrt{V_r}$ and $\varphi = \max_{r\in[R]} \sqrt{\frac{E}{6V_r}}$. By $\nabla f(\omega^r) = \frac{1}{K}\sum_{k=1}^{K} \nabla_{\Tilde{\omega}} F_k(\Tilde{\omega}_{k,r}) $, it holds:
	\begin{align}
		\frac{1}{R} & \sum_{r=0}^{R-1} \mathbb{E} \Big\| \frac{1}{K}\sum_{k=1}^{K} \nabla_{\Tilde{\omega}} F_k(\Tilde{\omega}_{k,r})\Big\|^2 \notag \\
		& \leq \frac{\phi}{\sqrt{ER}}\big(f_0 - f^*\big) + \frac{\varphi}{\sqrt{R}}\Big (2H^2+\frac{\sigma^2}{KE} \Big) \notag \\
		& \quad + \frac{5}{24R}\Big(1+\frac{2\varphi}{\sqrt{R}}\Big)(\sigma^2+6EB^2+18EH^2).
	\end{align}
Thus, Theorem \ref{convergence} can be proven.
\end{proof}

\section{Experiment Evaluation}
\label{experiment}
We conduct extensive experiments on the classic datasets and models for image classification and next-word prediction tasks, aiming at answering the following questions: 
\begin{itemize}
	\setlength{\itemsep}{0pt}
	\setlength{\parsep}{0pt}
	\setlength{\parskip}{0pt}
	\item \textbf{Q1:} Can FedLPS offer higher accuracy with fewer costs compared to baselines in dual-heterogeneous settings?
	\item \textbf{Q2:} Does FedLPS perform better in convergence?
	\item \textbf{Q3:} How do different non-IID and system-heterogeneous levels affect the performance of FedLPS?
	\item \textbf{Q4:} How do the learnable patterns and adaptive ratios obtained by P-UCBV affect performance, respectively? 
\end{itemize}

\subsection{Experimental Settings}
\paratitle{Datasets.} We adopt 5 benchmarks to evaluate the performance of FedLPS: MNIST \cite{MNIST}, CIFAR10, CIFAR100 \cite{CIFAR10}, Tiny-Imagenet \cite{tinyimagenet}, and Reddit \cite{leaf_benchmark_2018}. Reddit is a realistic federated dataset for the next-word prediction task, which contains a large number of English comment texts from real users \cite{leaf_benchmark_2018}. We adopt the top 100 users with more data as clients, where different clients have different sample sizes. Considering different speaking preferences, Reddit is inherently non-IID. The other datasets are widely used for image classification. MNIST and CIFAR10 involve 10-class images. For CIFAR100, there are 100-classe images. Tiny-Imagenet contains 100k color images from 200 classes. We utilize the pathological partition strategy \cite{DisPFL_2022} to obtain highly non-IID data, where each client is randomly assigned 2 classes for MNIST and CIFAR10, 10 classes for CIFAR100, and 20 classes for Tiny-Imagenet.

\begin{table*}[]
	\centering
	\renewcommand{\arraystretch}{1}
	\caption{Accuracy and FLOPs results of different methods.
	}
	\setlength{\tabcolsep}{2.6mm}{
		\begin{tabular}{l|cc|cc|cc|cc|cc}
			\toprule
			\multirow{4}{*}{Methods}& \multicolumn{2}{c|}{{MNIST}}  &  \multicolumn{2}{c|}{{CIFAR10}} & \multicolumn{2}{c|}{{CIFAR100}} & \multicolumn{2}{c|}{{Tiny-Imagenet}} & \multicolumn{2}{c}{{Reddit}} \\
			\cmidrule{2-11}
			& Acc & FLOPs & Acc & FLOPs & Acc & FLOPs & Acc & FLOPs & Acc & FLOPs \\
			& (\%) & (1e12) & (\%) & (1e12) & (\%) & (1e12) & (\%) & (1e12) & (\%) & (1e12) \\
			\midrule
			FedAvg & 88.15\scalebox{0.9}{$\pm$1.14} & 2.8 &  31.28\scalebox{0.9}{$\pm$1.24} & 825.4 &
			24.30\scalebox{0.9}{$\pm$0.53} & 4934.3 &
			5.21\scalebox{0.9}{$\pm$0.14} & 27081.0 & 23.39\scalebox{0.9}{$\pm$0.56} & 10.5 \\
			FedProx & 87.53\scalebox{0.9}{$\pm$1.16} & 2.8 & 31.12\scalebox{0.9}{$\pm$1.38} & 825.4 & 24.83\scalebox{0.9}{$\pm$0.52} & 4934.3 & 5.66\scalebox{0.9}{$\pm$0.12} & 27081.0 & 23.37\scalebox{0.9}{$\pm$0.57} & 10.5 \\
			Oort & 95.21\scalebox{0.9}{$\pm$0.42} & 2.8 & 33.62\scalebox{0.9}{$\pm$0.97} & 825.4 & 26.56\scalebox{0.9}{$\pm$0.39} & 4934.3 & 4.93\scalebox{0.9}{$\pm$0.12} & 26539.4 &
			23.91\scalebox{0.9}{$\pm$0.06} & 10.5 \\
			REFL & 94.67\scalebox{0.9}{$\pm$0.33} & 1.3 & 33.10\scalebox{0.9}{$\pm$1.68} & 642.1 & 24.87\scalebox{0.9}{$\pm$0.64} & 2950.7 & 5.03\scalebox{0.9}{$\pm$0.10} & 15436.2 &
			24.45\scalebox{0.9}{$\pm$0.07} & 8.8 \\
			PruneFL & 90.90\scalebox{0.9}{$\pm$0.34} & 2.2 & 32.20\scalebox{0.9}{$\pm$1.66} & 793.3 & 26.89\scalebox{0.9}{$\pm$0.46} & 4673.8  &	5.03\scalebox{0.9}{$\pm$0.13} & 26076.9 & 23.60\scalebox{0.9}{$\pm$0.19} & 9.0 \\
			CS & 92.44\scalebox{0.9}{$\pm$1.16} & 1.4 & 28.09\scalebox{0.9}{$\pm$0.84} & 417.0 & 27.76\scalebox{0.9}{$\pm$0.61} & 2493.7 & 6.02\scalebox{0.9}{$\pm$0.10} & 13686.2 & 21.69\scalebox{0.9}{$\pm$0.24} & 5.3 \\
			\midrule
			eFD & 93.31\scalebox{0.9}{$\pm$0.91} & 1.1 & 28.44\scalebox{0.9}{$\pm$0.98} & 380.1 & 22.74\scalebox{0.9}{$\pm$0.45} & 1989.5 &	3.85\scalebox{0.9}{$\pm$0.11} & 10337.5 & 23.20\scalebox{0.9}{$\pm$0.04} & 4.6 \\
			Fjord & 92.56\scalebox{0.9}{$\pm$1.03} & 1.1 & 28.04\scalebox{0.9}{$\pm$1.27} & 380.1 & 25.83\scalebox{0.9}{$\pm$0.77} & 1989.5 & 3.77\scalebox{0.9}{$\pm$0.08} & 10337.5 & 23.09\scalebox{0.9}{$\pm$0.56} & 4.6 \\
			HeteroFL & 90.76\scalebox{0.9}{$\pm$1.05} & 1.1 & 33.22\scalebox{0.9}{$\pm$1.73} & 380.1 & 26.34\scalebox{0.9}{$\pm$0.77} & 1989.5 & 4.65\scalebox{0.9}{$\pm$0.11} & 10337.5& 23.38\scalebox{0.9}{$\pm$0.11} & 4.6 \\
			FedRolex & 93.73\scalebox{0.9}{$\pm$1.44} & 1.1 & 33.59\scalebox{0.9}{$\pm$1.91} & 380.1 & 26.28\scalebox{0.9}{$\pm$1.02} & 1989.5 & 4.82\scalebox{0.9}{$\pm$0.12} & 10337.5 & 23.84\scalebox{0.9}{$\pm$0.25} & 4.6 \\
			FedMP & 91.63\scalebox{0.9}{$\pm$0.90} & 1.7 & 28.99\scalebox{0.9}{$\pm$0.78} & 532.0 & 30.36\scalebox{0.9}{$\pm$0.95} & 2373.5 & 5.40\scalebox{0.9}{$\pm$0.16} & 16019.2 & 23.57\scalebox{0.9}{$\pm$0.87} & 6.3 \\
			DepthFL & 95.44\scalebox{0.9}{$\pm$0.53} & 2.2 & 29.72\scalebox{0.9}{$\pm$1.06} & 725.8 & 20.41\scalebox{0.9}{$\pm$0.52} & 3431.6 &3.83\scalebox{0.9}{$\pm$0.13} & 23565.1 & 23.11\scalebox{0.9}{$\pm$0.16} & 4.9 \\
			\midrule
			Ditto & 92.00\scalebox{0.9}{$\pm$0.30} & 5.6 & 82.60\scalebox{0.9}{$\pm$0.29} & 1650.8 & 49.01\scalebox{0.9}{$\pm$0.19} & 9868.6 & 11.92\scalebox{0.9}{$\pm$0.04} & 54162.0 & 23.97\scalebox{0.9}{$\pm$0.24} & 21.0 \\
			FedPer & 93.45\scalebox{0.9}{$\pm$0.07} & 2.8 & 79.76\scalebox{0.9}{$\pm$0.15} & 825.4 & 60.43\scalebox{0.9}{$\pm$0.08} & 4934.3 & 14.96\scalebox{0.9}{$\pm$0.10} & 27081.0 & 24.64\scalebox{0.9}{$\pm$0.06} & 10.5 \\ 
			FedRep & 91.92\scalebox{0.9}{$\pm$0.30} & 2.8 & 80.87\scalebox{0.9}{$\pm$0.09} & 825.4 & 45.36\scalebox{0.9}{$\pm$0.53} & 4934.3 & 15.18\scalebox{0.9}{$\pm$0.06} & 27081.0 & 24.43\scalebox{0.9}{$\pm$0.07} & 10.5 \\
			Per-FedAvg & 93.46\scalebox{0.9}{$\pm$0.33} & 2.8 & 85.00\scalebox{0.9}{$\pm$0.09} & 660.3 & 58.57\scalebox{0.9}{$\pm$1.57} & 3256.6 & 16.60\scalebox{0.9}{$\pm$0.51} & 18956.7 & 24.65\scalebox{0.9}{$\pm$0.42} & 9.7 \\
			LotteryFL &  93.37\scalebox{0.9}{$\pm$1.28} & 1.8 & 75.71\scalebox{0.9}{$\pm$1.00} & 656.7 & 44.83\scalebox{0.9}{$\pm$0.58} & 3276.9 & 11.11\scalebox{0.9}{$\pm$0.21} & 16454.3 & 24.33\scalebox{0.9}{$\pm$0.02} & 6.2 \\
			Hermes &  94.05\scalebox{0.9}{$\pm$0.60} & 1.7 & 81.07\scalebox{0.9}{$\pm$1.58} & 565.9 & 47.09\scalebox{0.9}{$\pm$0.81} & 3547.7 & 11.52\scalebox{0.9}{$\pm$0.10} & 17168.1 & 24.89\scalebox{0.9}{$\pm$0.10} & 6.1 \\
			FedSpa &  93.94\scalebox{0.9}{$\pm$0.85} & 1.4 & 78.83\scalebox{0.9}{$\pm$2.41} & 522.9 & 55.76\scalebox{0.9}{$\pm$1.02} & 3057.5 & 9.20\scalebox{0.9}{$\pm$0.12} & 16229.8 & 24.62\scalebox{0.9}{$\pm$0.04} & 5.3  \\
			FedP3 & 92.41\scalebox{0.9}{$\pm$0.34} & 1.0 & 78.46\scalebox{0.9}{$\pm$0.16} & 320.1 & 54.20\scalebox{0.9}{$\pm$0.19} & 1951.6 & 17.69\scalebox{0.9}{$\pm$0.05} & 13093.1 & 23.91\scalebox{0.9}{$\pm$0.09} & 4.6 \\
			\midrule
			\textbf{FedLPS} & \textbf{96.77\scalebox{0.9}{$\pm$0.11}} & \textbf{0.8} & \textbf{87.43\scalebox{0.9}{$\pm$0.19}} & \textbf{268.5} & \textbf{62.80\scalebox{0.9}{$\pm$0.05}} & \textbf{1706.7} & \textbf{21.48\scalebox{0.9}{$\pm$0.12}} & \textbf{9307.3} & \textbf{26.17\scalebox{0.9}{$\pm$0.04}} & \textbf{4.2} \\ 
			\bottomrule
	\end{tabular}}
	\label{main_results}
	\vskip -0.05in
\end{table*}

\paratitle{Models.} A CNN model with two convolutional layers is used for MNIST, and the VGG11 \cite{VGG} is adopted as the backbone of CIFAR10. CIFAR100 and Tiny-Imagenet are trained on VGG13 and VGG16. We employ a RNN model with two LSTM layers and a softmax layer \cite{leaf_benchmark_2018} for Reddit.

\paratitle{Baselines.} We compare against various FL frameworks, including
(1) \emph{conventional FL}, (2) \emph{heterogeneous sparse-training FL}, and (3) \emph{personalized FL}. For \underline{\emph{conventional FL}} frameworks with the same size models on all clients, we consider that: 
\begin{itemize}
	\setlength{\itemsep}{0pt}
	\setlength{\parsep}{0pt}
	\setlength{\parskip}{0pt}
	\item \emph{FedAvg} \cite{fedavg_2017} and \emph{FedProx} \cite{FedProx_2020} are classic FL frameworks, which require clients to locally train dense models and upload all updates for average aggregation. 
	\item \emph{Oort} \cite{Oort_2021} and \emph{REFL} \cite{REFL_2023} explore intelligent client selection in heterogeneous FL. Oort ignores local data diversity, while REFL alleviates the diversity issue by diversity measure and staleness-aware aggregation, where the stale updates negatively impact accuracy and convergence.
	\item \emph{PruneFL} \cite{PruneFL_2022} and \emph{CS} \cite{CS_2023} are state-of-the-art FL sparsification methods, where clients inherit the same sparse ratio and train submodels with the same sizes. PruneFL requires a powerful client for initial dense model sparsification and then distributes the sparse model to all clients for joint learning. CS applies unstructured sparsification in FL, limited by the specialized hardware requirement.
\end{itemize}
In terms of \underline{\emph{heterogeneous sparse-training FL}} frameworks,
\begin{itemize}
	\item \emph{Fjord} \cite{Fjord_2021}, \emph{HeteroFL} \cite{HeteroFL_2021}, and \emph{FedRolex} \cite{FedRolex_2022} directly select sparse ratios based on local capabilities and prune model units in an ordered manner.
	\item \emph{DepthFL} \cite{DepthFL_2023} tailors local models by removing the deepest layers under the resource-based sparse ratios.
	\item \emph{FedMP} \cite{FedMP_2022} explores extended UCB to select sparse ratio and then prunes model units based on magnitude.
\end{itemize}
Among \underline{\emph{personalized FL}} methods, we compare that: 
\begin{itemize}
	\item \emph{FedPer} \cite{FedPer_2019} and \emph{FedRep} \cite{FedRep_2021} view the last layers of local models as personalized modules that are not uploaded.
	\item \emph{Ditto} \cite{ditto_PFL_2021} introduces regularization into the local training for robust personalization.
	\item \emph{PerFedAvg} \cite{Per_FedAvg_2020} studies the personalized variant of FedAvg within model-agnostic meta-learning.
	\item \emph{LotteryFL} \cite{LotteryFL_2021} and \emph{Hermes} \cite{Hermes_2021} gradually decline sparse ratios from 1 at a fixed rate and prune a fixed number of low-magnitude weights to customize local sparse models.
	\item \emph{FedSpa} \cite{FedSpa_2022} introduces dynamic sparse training into FL to learn an always-sparse personalized model with a uniform sparse ratio for each client.
	\item \emph{FedP3} \cite{FedP3_2024} integrates global and local pruning strategies and allows personalization based on the client resource constraints, without involving pattern optimization.
\end{itemize}

\paratitle{Configurations.} For heterogeneous clients, we consider five capability levels $z_k \in \{1, 1/2, 1/4, 1/8, 1/16\}$ and uniformly sample from possible levels for $K$ clients, referring to \cite{HeteroFL_2021}. We set the optimal device (\emph{i.e.}, $z_k=1$) with Adreno 630, which has the computation capability of 727G floating-point operations per second \cite{Adreno630}. During the training, the local available resources can dynamically change, since users also have other irregular task requirements that may bring the changes of available power. The total number of clients is $K=100$ for MNIST and Reddit, and $K=50$ is set for other datasets. The number of communication rounds is $R=100$. In each round, the server randomly selects 10 clients. During local training, the batch size is set to 20. SGD optimizer is adopted with $0.1$ learning rate for image classification tasks and $8$ for next-word prediction, where gradient clip is involved as referring to \cite{leaf_benchmark_2018}. For (\ref{our_loss}), we set $\mu = 1$ and $\lambda = 1$. The utility function used in (\ref{reward}) is defined to be $U(x) = 10- \frac{20}{1+e^{(0.35x)}}$. Our code is available at \url{https://github.com/sunnyxuejj/FedLPS}.

\paratitle{Evaluation Metrics.} We adopt test accuracy and FLOPs as the main evaluation metrics. For FLOPs, many FL works \cite{DisPFL_2022, PruneFL_2022, FedSpa_2022} use it to characterize local computation costs, where fewer FLOPs mean less training overhead. Generally, FLOPs and running time are positively correlated in certain settings \cite{PruneFL_2022, flops_time_2021}. The total time cost is also evaluated in our experiments.

\begin{figure*}[htbp]
	\centering
	\includegraphics[width=0.93\textwidth]{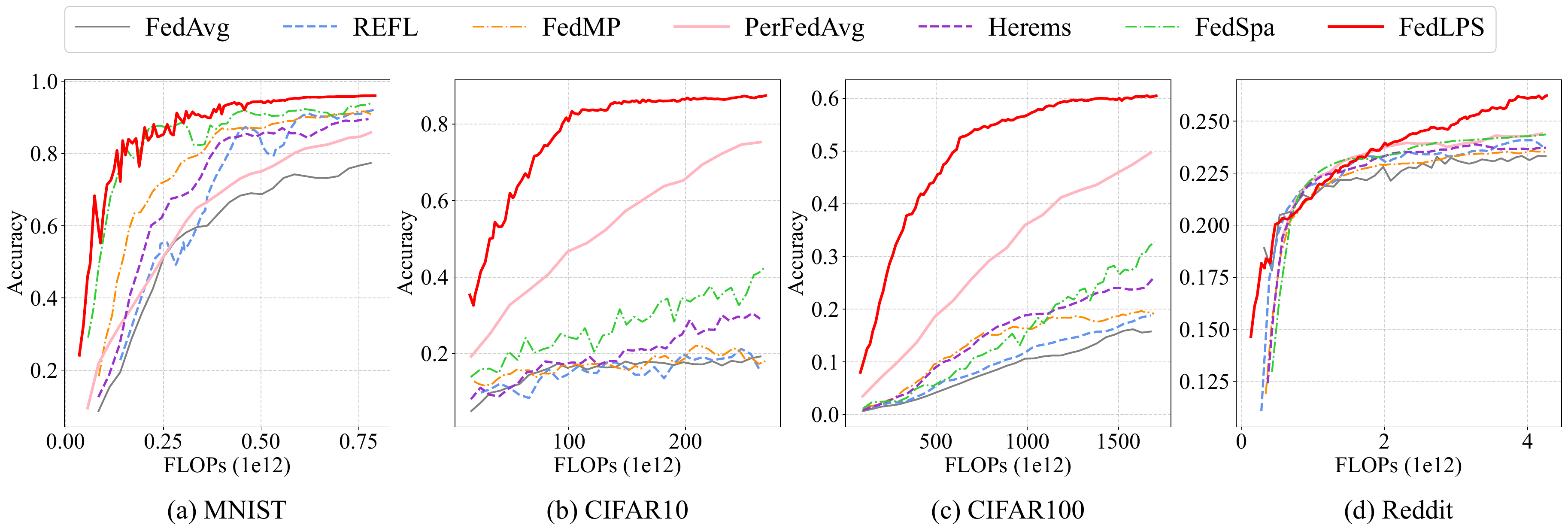}
	\caption{Test accuracy versus FLOP computation costs on four datasets.}
	\label{acc_flops}
	\vskip -0.1in
\end{figure*}

\begin{figure*}[htbp]
	\centering
	\includegraphics[width=0.93\textwidth]{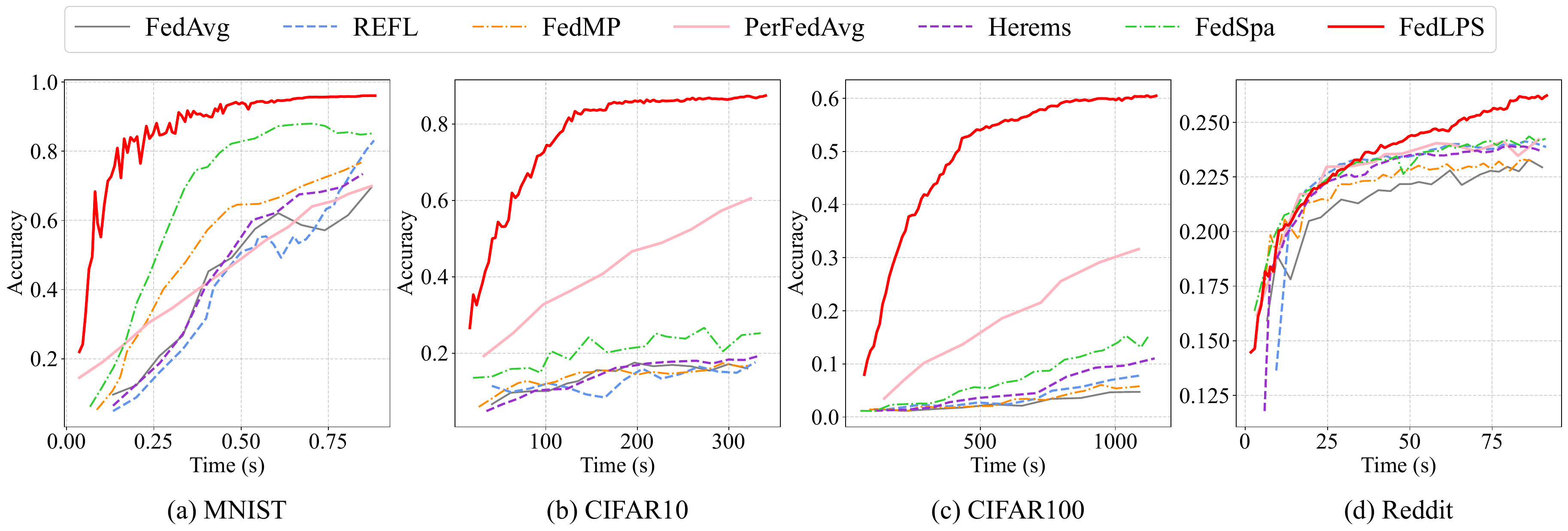}
	\caption{Test accuracy versus running time on four datasets.}
	\label{acc_times}
	\vskip -0.1in
\end{figure*}

\subsection{Performance Comparison (Q1)}

Table \ref{main_results} summarizes experimental results, where \emph{Acc} means the average accuracy of all clients on local test data, and \emph{FLOPs} denote total floating operations of all clients during the FL process. It is remarkable that FedLPS consistently achieves the start-of-the-art accuracy with minimal computation costs. 

First, conventional FL frameworks FedAvg and FedProx generally perform poorly on all datasets. Although Oort and REFL mitigate the heterogeneity bottleneck by adaptive client selection, they still suffer from accuracy degradation in non-IID settings. Compared to REFL, FedLPS achieves 1.72\%-54.33\% accuracy gain with 38.46\%-58.18\% FLOP reduction. Besides, recent sparse works (based on conventional FL) PruneFL and CS also reduce local FLOPs, but still show lower accuracy. Compared to CS, FedLPS improves accuracy by 4.33\%-59.34\% while reducing 20.75\%-42.86\% FLOP costs.

Second, heterogeneous sparse FL methods significantly outperform conventional FL. A major merit is that they assign submodels with different sizes based on client-side computing power. However, they finally deploy a shared inference model on all clients, resulting in lower accuracy in non-IID settings. In particular, FedLPS outperforms Hermes in test accuracy and training costs, which enhances accuracy by 1.28\%-15.71\% and reduces computation costs by more than 30\%. Compared to the advanced FedRolex, FedLPS achieves 2.33\%-53.84\% accuracy gains and saves up to 29\% FLOP computation costs.

Furthermore, personalized FL methods tailor a local model for each client to fit non-IID data and yield better inference accuracy. Combined with sparsification techniques, LotteryFL, Herems, FedSpa, and FedP3 learn personalized sparse models with fewer computation costs. Finally, FedLPS outperforms state-of-the-art personalized FL methods in inference accuracy and computation overhead on all datasets. Specifically, FedLPS provides 1.28\%-4.23\% accuracy gains with up to 60\% reduction of FLOPs.

\begin{figure}
	\begin{minipage}[t]{0.48\linewidth}
		\centering
		\centerline{\includegraphics[width=\columnwidth]{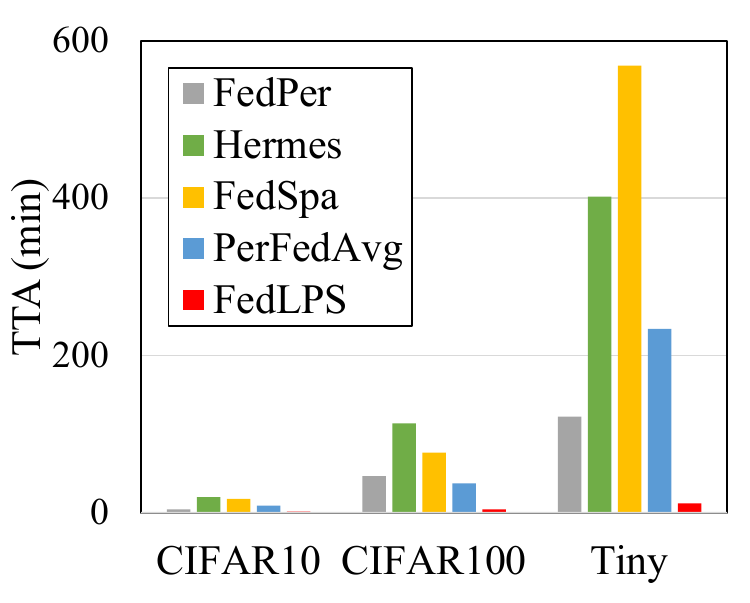}}
		\caption{TTA on CIFAR10, CIFAR100, and Tiny-Imagenet.}
		\label{TTA}
	\end{minipage}
	\hspace{0.04in}
	\begin{minipage}[t]{0.48\linewidth}
		\centering
		\centerline{\includegraphics[width=\columnwidth]{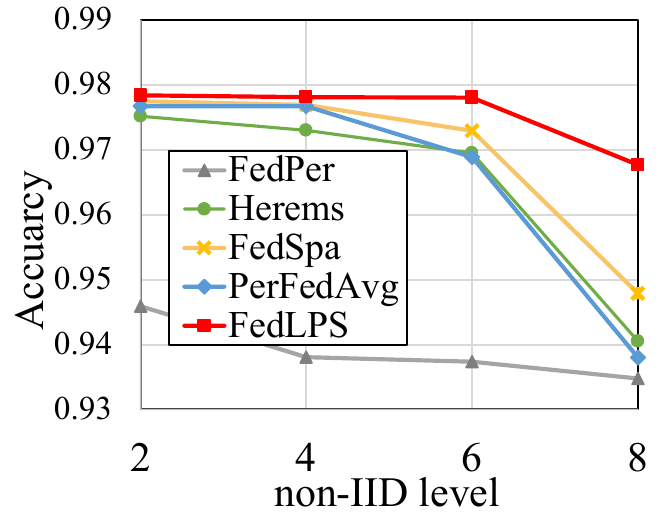}}
		\caption{ Accuracy versus non-IID levels on MNIST.}
		\label{mnist_diff_nonIID}
	\end{minipage}%
	\vskip -0.15in
\end{figure}

\subsection{Convergence and Running Time Evaluation (Q1 and Q2)}
To evaluate the convergence of FedLPS, we report the test accuracy varying with total FLOPs and running time. As shown in Figure \ref{acc_flops}, FedLPS generally offers higher accuracy under the same computation costs. Noticeably, within $1500 \times 10^{12}$ FLOPs, the accuracy of FedLPS in the last three rounds over CIFAR100 is 59.86\%, while PerFedAvg and FedSpa provide 47.46\% and 29.27\% test accuracy. Besides, we observe that FedLPS can quickly converge to a target accuracy with less time than other methods, as shown in Figure \ref{acc_times}. 

We adopt Time-To-Accuracy (TTA) to represent the running time required to reach target accuracy, motivated by \cite{FedBIAD_2023}. As shown in Figure \ref{TTA}, FedLPS consistently takes the shortest time to reach the target accuracy. For CIFAR10, FedLPS achieves 70\% test accuracy after 90.74s, which reduces more than 68\% (vs. 290.83s) running time compared to baselines. For CIFAR100, FedLPS reaches 40\% accuracy by 281.81s, which provides more than 80\% (vs. 2261.89s) time saving.

\begin{figure}
	\centering
	\includegraphics[width=0.93\columnwidth]{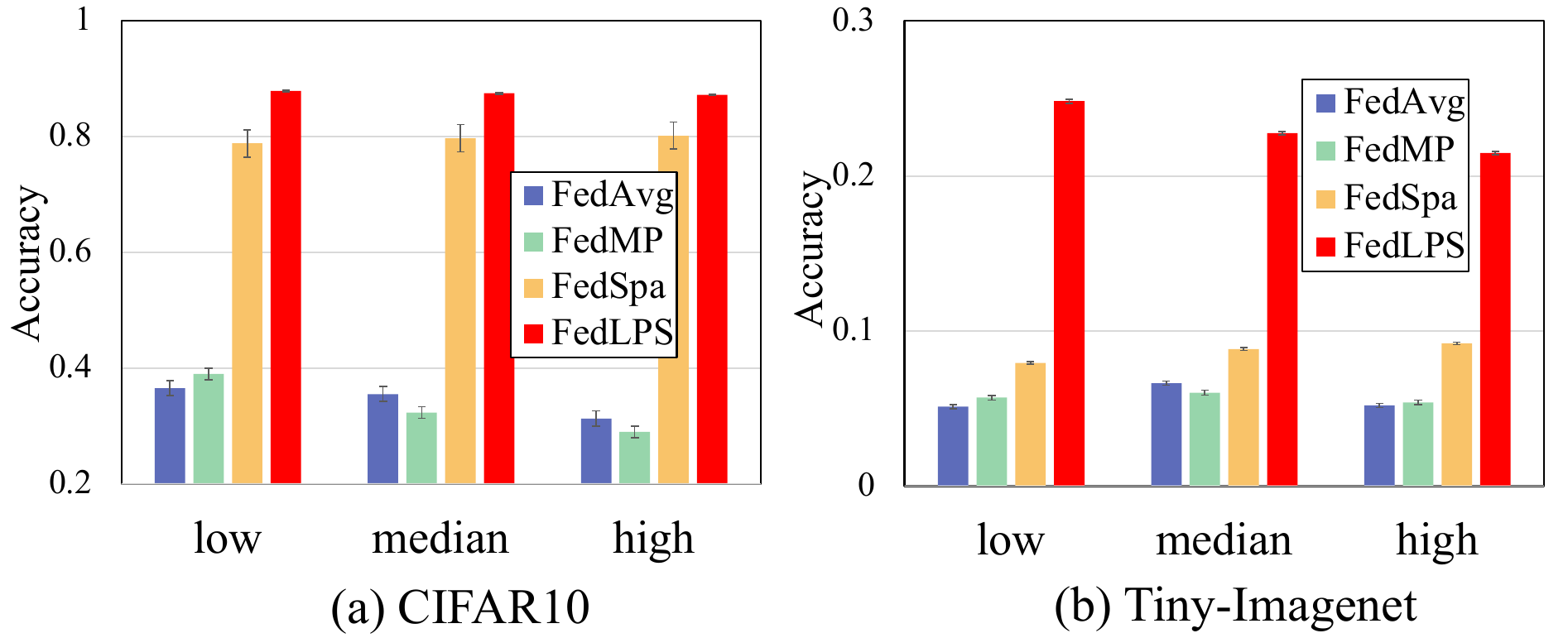}
	\caption{Test accuracy under different system heterogeneity levels on CIFAR10 and Tiny-Imagenet.}
	\label{diff_herero_acc}
	\vskip -0.1in
\end{figure}
\begin{figure}
	\centering
	\includegraphics[width=0.93\columnwidth]{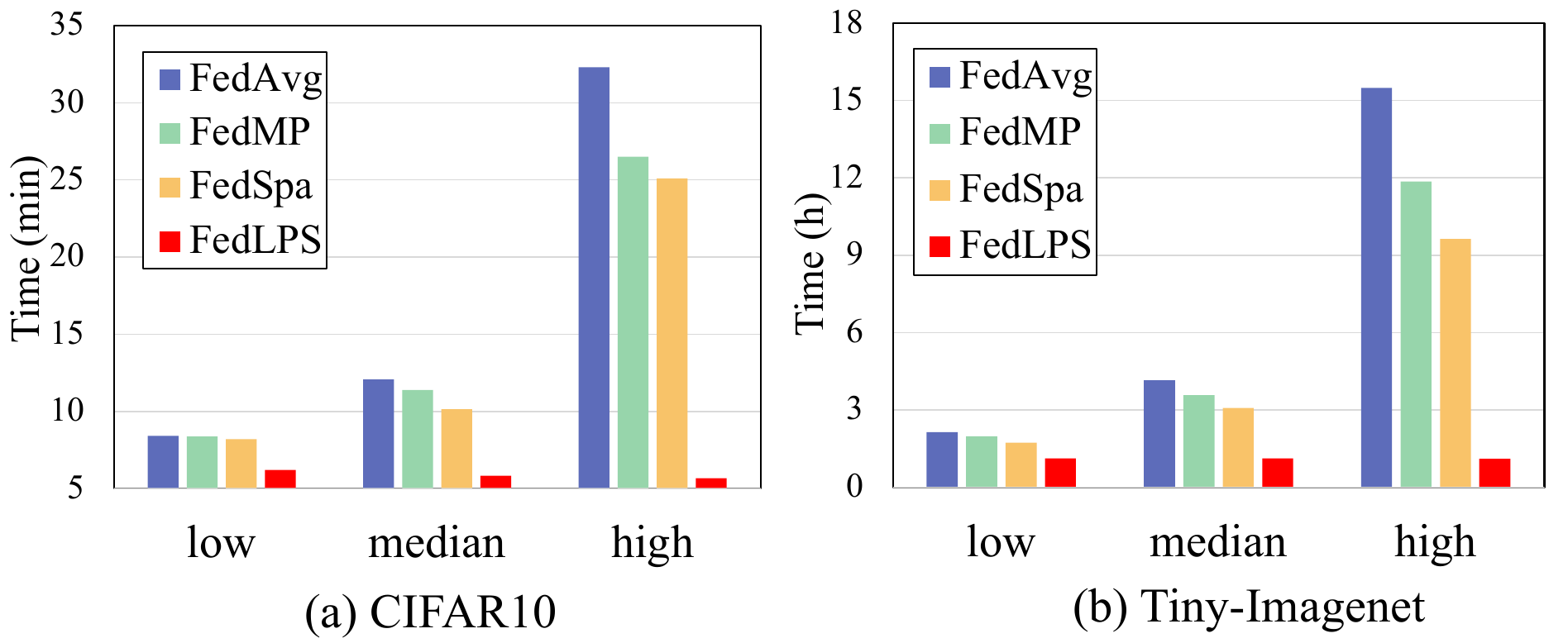}
	\caption{Time results under different system heterogeneity levels on CIFAR10 and Tiny-Imagenet.}
	\label{diff_herero_time}
	\vskip -0.1in
\end{figure}

\subsection{Effect of Heterogeneity (Q3)}

We also conduct experiments to evaluate FedLPS and several personalized baselines under different non-IID levels. The results on MNIST are shown in Figure \ref{mnist_diff_nonIID}, where the horizontal axis $x$ implies that each client lacks $x$ classes training data and the larger $x$ indicates the higher non-IID level. As the non-IID level rises, the average test accuracy of different methods gradually declines. Apparently, FedLPS always outperforms baselines under different non-IID levels and the accuracy advantages are more significant on higher non-IID levels. 

Furthermore, the effect of various heterogeneity levels is also explored. We consider three heterogeneity levels: low ($z_k \in\{1, 1/2\}$), median ($z_k \in\{1, 1/2, 1/4\}$), and high ($z_k \in\{1, 1/2, 1/4, 1/8, 1/16\}$). As shown in Figure \ref{diff_herero_acc} and Figure \ref{diff_herero_time}, FedLPS consistently keeps higher accuracy and greater time advantage, which enhances accuracy by 7.08\%-55.95\% and 12.28\%-19.69\% in CIFAR10 and Tiny-Imagenet compared to the other methods.
From low to high, the running time of three baselines increases accordingly, while FedLPS can basically remain stable and shorten 24.46\%-88.46\% running time, demonstrating the reliability and efficiency of FedLPS.

\begin{table}[]
	\centering
	\renewcommand{\arraystretch}{1.1}
	\caption{The results of ablation experiments.
	}
	\setlength{\tabcolsep}{1mm}{
		\begin{tabular}{l|cc|cc|cc}
			\toprule
			\multirow{4}{*}{Methods}& \multicolumn{2}{c|}{{MNIST}}  &  \multicolumn{2}{c|}{{CIFAR10}} & \multicolumn{2}{c}{{Reddit}} \\
			\cmidrule{2-7}
			& Acc & FLOPs & Acc & FLOPs & Acc & FLOPs \\
			& (\%) & (1e12) & (\%) & (1e12) & (\%) & (1e12) \\
			\midrule
			FLST & 94.76\scalebox{0.6}{$\pm$0.06} & 1.4 & 87.32\scalebox{0.6}{$\pm$0.26} & 412.9 & 24.85\scalebox{0.6}{$\pm$0.07} & 5.3 \\
			\midrule
			RCR-Fix & 94.95\scalebox{0.6}{$\pm$0.20} & 1.1 &  78.84\scalebox{0.6}{$\pm$0.10} & 394.3& 23.39\scalebox{0.6}{$\pm$0.56} & 5.1 \\
			\textbf{P-UCBV-Fix} & \textbf{95.72\scalebox{0.6}{$\pm$0.26}} & \textbf{0.9} & \textbf{80.71\scalebox{0.6}{$\pm$0.44}} & \textbf{325.0} & \textbf{23.71\scalebox{0.6}{$\pm$0.08}} & \textbf{3.7} \\
			\midrule
			RCR-Dyn & 95.99\scalebox{0.6}{$\pm$0.12} & 1.1 & 86.06\scalebox{0.6}{$\pm$0.12} & 380.1 & 25.20\scalebox{0.6}{$\pm$0.06} & 4.6 \\
			\textbf{P-UCBV-Dyn} & \textbf{96.77\scalebox{0.6}{$\pm$0.11}} & \textbf{0.8} & \textbf{87.43\scalebox{0.6}{$\pm$0.19}} & \textbf{268.5} & \textbf{26.17\scalebox{0.6}{$\pm$0.04}} & \textbf{4.2} \\
			\bottomrule
	\end{tabular}}
	\label{ablation_results}
	\vskip -0.05in
\end{table}

\begin{figure}
	\centering
	\subfloat[Test accuracy versus sparse ratio by different sparsification strategies on MNIST (left) and Reddit (right).]
	{
		\includegraphics[width=0.98\columnwidth]{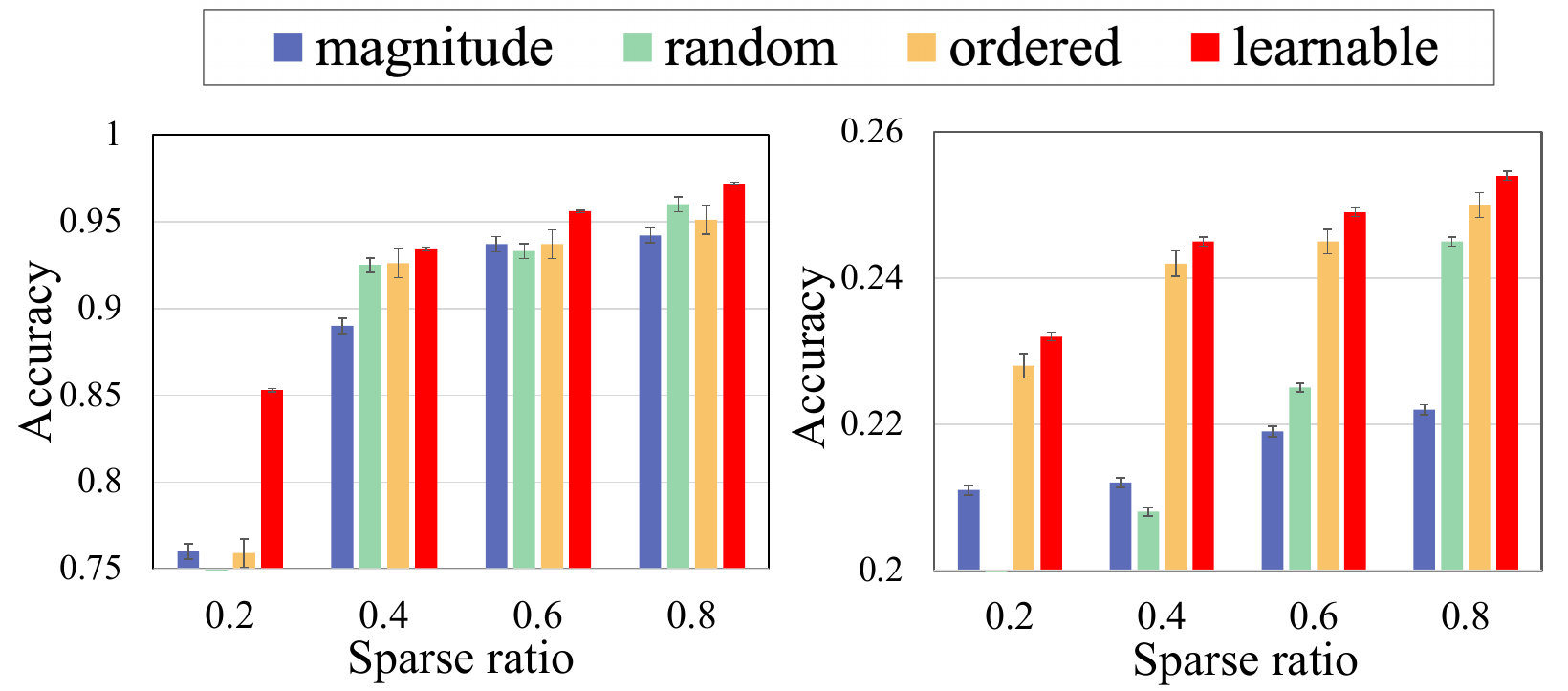}
		\label{diff_ratio}
	} 
	\quad
	\subfloat[Time results of our learnable sparsification on MNIST (left) and Reddit.]
	{
		\includegraphics[width=0.98\columnwidth]{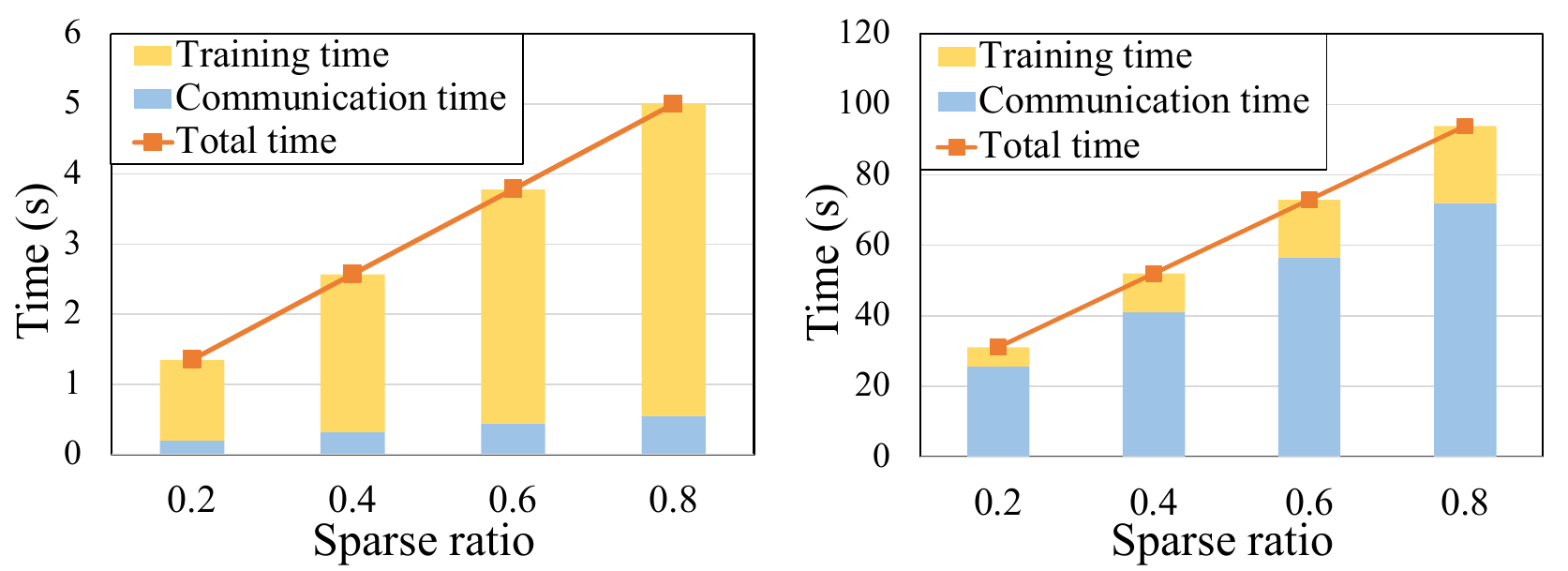}
		\label{diff_ratio_time}
	}
	\caption{Accuracy and time results under different sparse ratios.}
	\label{diff_ratio_results}
	\vskip -0.1in
\end{figure}

\subsection{Ablation Study (Q4)}
We further investigate the effect of learnable personalized patterns on a fixed sparse ratio and verify the advantage of adaptive ratios decided by P-UCBV. First, we set a fixed ratio $s_k=0.5$ for all clients and learn personalized sparse patterns through our Learnable Sparse Training (\emph{i.e.}, FLST). The results in Table \ref{ablation_results} show that FLST generally offers higher accuracy than the start-of-the-art sparse FL works CS and FedSpa (see Table \ref{main_results}, where CS and FedSpa also adopt the same ratio $s_k=0.5$) under the same FLOPs, revealing the vital contribution of our learnable personalized pattern. To quantitatively verify the advantages of learnable sparsification, we compare against existing heuristic sparsifications (including random, ordered, and magnitude-based strategies) under different sparse ratios, where the sparse ratio is set to $s_k\in\{0.2, 0.4, 0.6, 0.8\}$ for all clients. As shown in Figure  \ref{diff_ratio}, FedLPS consistently outperforms other methods under various ratio settings, demonstrating the reliability of our learnable strategy. Besides, we observe from Figure \ref{diff_ratio_time} that as the ratio increases, the test accuracy is enhanced while running time also increases, indicating the significance of the ratio selection.

To verify the advantages of P-UCBV in ratio decisions, we compare it with a rigid Resourced-Controlled Ratio (RCR) rule under fixed and dynamic system heterogeneity levels. Sparse ratios are directly set according to local computation capabilities in RCR, which is used in \cite{FedRolex_2022, Fjord_2021}, and \cite{HeteroFL_2021}. As shown in Table \ref{ablation_results}, RCR shows lower accuracy and more FLOPs than P-UCBV, which indicates that our P-UCBV significantly reduces training costs while guaranteeing inference accuracy.

\section{Related Work}

\paratitle{Model Sparsification.} Modern DNNs typically contain millions of parameters, requiring high computing power for training \cite{model_computing_2018}. For resource-limited edge devices, it is infeasible to completely train dense models. Hence, for lightening models, sparsification has attracted widespread attention \cite{Network_Pruning_2019, unstructured_sparse_2020, chen_prune_2022, Chen2023SparseMA}, which can be characterized as unstructured and structured. \emph{Unstructured sparsification} simply zeros out partial parameters without considering structures of DNNs, where sparse parameter matrices are too irregular to accelerate training on commodity hardware \cite{unstrcture_speed_2016}. Although several works \cite{FedDST_2022, CS_2023} apply unstructured sparsification into FL, they require specialized hardware/libraries on clients and cannot adapt to diverse configurations \cite{FedMP_2022}. On the contrary, \emph{structured sparsification} \cite{unstrcture_speed_2016, structure_sparsity_2018} takes DNN structures (\emph{i.e.}, units) into account and implements structure-wise pruning. The remaining part after structured sparsification is regular and can be viewed as a submodel, thereby enabling training speed-ups on general-purpose hardware. Federated Dropout \cite{FD_2018} (FD) first introduces structured sparsification into FL for MEC, but it adopts uniform sparse ratios and random sparse patterns, ignoring ratio determinations and incurring precision sacrifices.

\paratitle{System Heterogeneity.} System heterogeneity is one of the primary challenges of edge data management and computing. Conventional FL requires global and local models to share the same architecture \cite{client_select_2019}, where low-tier clients are excluded from training, leading to training bias and accuracy degradation \cite{FL_survey}. Knowledge distillation \cite{FedMD_2019, KD_HFL_2021} supports training different models, but requires public datasets for fine-tuning, which violates the data localization and adds extra overhead \cite{HeteroFL_2021}. Several works explore sparsification to extract different submodels. The common idea is to assign sparse models with different sizes based on local resources, where sparse ratios determine model sizes. FedDrop \cite{FedDrop_2022} adjusts sparse ratios depending on bandwidth limitation, while Fjord \cite{Fjord_2021} selects ratios based on local computing power. They focus on the association between sparse ratios and resources without involving local data, resulting in unguaranteed performance in non-IID settings. For sparse pattern selection, existing works can be roughly categorized into depth and width scaling \cite{AdaptiveFL_2024}. Depth scaling \cite{InclusiveFL_2022, DepthFL_2023} constructs heterogeneous models by removing several deepest layers. Such layer-wise sparsification cannot support fine-grained adjustments of patterns and is more prone to larger accuracy fluctuations. In contrast, width scaling tailors sparse models by pruning neurons/channels. Fjord \cite{Fjord_2021}, HeteroFL \cite{HeteroFL_2021}, and FedRolex \cite{FedRolex_2022} tailor sparse models in an ordered manner, meaning that adjacent units are dropped out first in these works. While FedMP \cite{FedMP_2022} prunes the units with lower magnitude. The above ordered and magnitude-based sparsifications are strongly heuristic \cite{STR2022} and cannot flexibly customize sparse patterns. Finally, a shared model is obtained for inference on all clients, failing to generalize well on clients with non-IID data \cite{PFL_survey_2021}.

\paratitle{Sparsification with Statistical Heterogeneity.} The non-IID problem is common in EDM and MEC due to the distinct preferences of edge users \cite{non-IID_problem_2023}. To tackle the problem, several FL works propose to embed personalized information into client-side sparse models. Hermes \cite{Hermes_2021} and LotteryFL \cite{LotteryFL_2021} heuristically prune lower-magnitude parameters or units to tailor a personalized sparse model for each client. However, the magnitude cannot accurately reflect the importance of model units over local data, resulting in accuracy degradation. Moreover, they take dense-to-sparse rules, gradually decreasing sparse ratios at a fixed rate after reaching preset accuracy thresholds. In this way, weaker clients are still required to train larger scaling submodels and suffer from unbearable training loads, which delays the FL progress and even induces stragglers. Sparse-to-sparse techniques are explored in FedSpa \cite{FedSpa_2022}, where each client shifts the submodel regularly by the iterative cuts of the lowest magnitude parameters and random growth of other connections during local training, where heuristic magnitude-based pattern adjustment also cannot guarantee accuracy. Besides, FedSpa restricts clients to hold the same and constant sparse ratio, unable to be applied in system-heterogeneous MEC scenarios. Although the state-of-the-art FedP3 \cite{FedP3_2024} considers system heterogeneity, it still adopts heuristic sparse patterns (\emph{i.e.,} uniform and ordered dropout).



\section{Conclusion}

There has been a growing interest in distributed data management at network edges to utilize data in a real-time and privacy-preserving way. In this paper, we focus on the challenges of statistical and system heterogeneity in such edge scenarios, and propose Learnable Personalized Sparsification for heterogeneous Federated learning (FedLPS), which facilitates the learnable customization of sparse patterns and sparse ratios to boost model performance and computing efficiency. We integrate the importance of model units into local loss and learn the importance on local data to tailor importance-based sparse patterns with minimal heuristics, which can accurately extract personalized data features. Furthermore, P-UCBV learns the additive effect of diverse capabilities and non-IID data via the MAB modeling and superimposed feedback designing to adaptively determine sparse ratios for the trade-off between capability adaption and accuracy improvement. We conduct extensive experiments on typical datasets with various models, and the results show that FedLPS significantly outperforms the state-of-the-art method in test accuracy and training efficiency.

Our learnable idea inspires that pending indicators might be learned together with models in a non-heuristic way, which provides a future direction for learnable hyperparameter customization in edge data management and computing.

\section*{Acknowledgements}
This article is supported in part by the National Key Research and Development Program of China (No. 2021YFB2900102); in part by the National Natural Science Foundation of China (NSFC) under Grants 62472410 and 62072436. Additionally, Jingyuan Wang's work was partially supported by NSFC (No. 72171013, 72222022, and 72242101)

\footnotesize
\balance
\bibliographystyle{IEEEtran}
\bibliography{FedLPS_cr.bib}

\end{document}